\documentclass[a4paper,fleqn]{cas-dc}

\usepackage[numbers]{natbib}

\def\tsc#1{\csdef{#1}{\textsc{\lowercase{#1}}\xspace}}
\tsc{WGM}
\tsc{QE}
\tsc{EP}
\tsc{PMS}
\tsc{BEC}
\tsc{DE}

\usepackage{amssymb}
\usepackage{amsthm}

\usepackage{adjustbox}
\usepackage{upquote}
\usepackage{enumerate}
\usepackage{amsmath}
\usepackage[utf8]{inputenc}
\usepackage{stackengine}
\usepackage{amsfonts}
\usepackage{mathtools}
\usepackage{arydshln}
\usepackage{footnote}
\usepackage{amssymb}
\usepackage{xcolor}
\usepackage{upquote}
\usepackage{graphicx}
\usepackage{multirow}
\usepackage{graphicx}
\usepackage{empheq}
\usepackage{tikz}
\usepackage{amssymb}
\usetikzlibrary{shapes,arrows}
  \tikzstyle{block} = [rectangle, draw, fill=green!30, 
    text width=8em, text centered, rounded corners, minimum height=7em]
\tikzstyle{block4} = [rectangle, draw, fill=green!30, 
    text width=12em, text centered, rounded corners, minimum height=8em]
\tikzstyle{block1} = [rectangle, draw, fill=red!20, 
    text width=8em, text centered, rounded corners, minimum height=7em,node distance=4cm]
\tikzstyle{block2} = [rectangle, draw=red, fill=yellow!35, 
    text width=9em, text centered, rounded corners, minimum height=3em,node distance=2cm]
\tikzstyle{block3} = [rectangle,draw, fill=yellow!25, 
    text width=15em, text centered, rounded corners, minimum height=4em,node distance=1cm]
\tikzstyle{line} = [draw,->, line width=0.5mm]
\tikzstyle{line1} = [draw,dotted,->, line width=0.7mm]
\tikzstyle{cloud} = [draw, ellipse,fill=red!35, text width=5.5em, text centered, node distance=3cm, minimum height=4.5em]
\tikzstyle{decision} = [diamond, draw, fill=blue!10, 
    text width=5em, text centered, node distance=4cm, inner sep=0pt]
\tikzstyle{line1} = [draw,dotted,->, line width=0.7mm]  

\definecolor{mypink}{rgb}{0.85, 0.25, 1}
\definecolor{mypink1}{rgb}{0.85, 0.25, 0.21}
\definecolor{mygreen}{rgb}{0.15, 0.65, 1}
\definecolor{mypink2}{RGB}{219, 48, 122}
\definecolor{mypink3}{cmyk}{0, 0.7808, 0.4429, 0.1412}
\definecolor{mygray}{gray}{0.6}

 \theoremstyle{plain}
\if@secthm
  \newtheorem{thm}{Theorem}[section]
  \@addtoreset{thm}{section}
\else
  \newtheorem{thm}{Theorem}
\fi

\newtheorem{lem}{Lemma}

\newtheorem{assum}{Assumption}

 \theoremstyle{definition}
\newtheorem{defn}{Definition}

\begin{document}
\let\WriteBookmarks\relax
\def\floatpagepagefraction{1}
\def\textpagefraction{.001}

\shorttitle{LMI-based Nonlinear FAUIO Approach for Fault Estimation}

\shortauthors{Shivaraj Mohite et~al.}

\title [mode = title]{Estimating the Faults/Attacks using a Fast Adaptive Unknown Input Observer: An Enhanced LMI Approach}                      
\tnotemark[1]



%
\author[1,2]{Shivaraj Mohite}






\affiliation[1]{organization={University of Lorraine, CRAN CNRS UMR 7039},
    addressline={186, rue de Lorraine}, 
    city={Cosnes et Romain},
    postcode={54400}, 
    country={France}}
\affiliation[2]{
    organization={
    Rhineland-Palatinate Technical University of Kaiserslautern-Landau, Germany},
    city={Kaiserslautern},
    postcode={67663}, 
    country={Germany}}

\author[3]{Adil Sheikh}[type=editor,
                        auid=000,bioid=1,
                        role=,
                        orcid=0000-0003-2678-351X]
                        \cormark[1]
\ead{adilsheikh1703@gmail.com}

\author[4]{S. R. Wagh}


\affiliation[3]{organization={Electrical Engineering Department,
St. Francis Institute of Technology (SFIT)},
    city={Mumbai},
    postcode={400103}, 
    country={India}}
\affiliation[4]{organization={Electrical Engineering Department, Veermata Jijabai Technological Institute (VJTI)},
    city={Mumbai},
    postcode={400019}, 
    country={India}}

\cortext[cor1]{Corresponding author}



\begin{abstract}
This paper deals with the problem of robust fault estimation for the Lipschitz nonlinear systems under the influence of sensor faults and actuator faults. In the proposed methodology, a descriptor system is formulated by augmenting sensor fault vectors with the states of the system.
A novel fast adaptive unknown input observer (FAUIO) structure is proposed for the simultaneous estimation of both faults and states of a class of nonlinear systems.
A new LMI condition is established by utilizing the $\mathcal{H}_\infty$ criterion to ensure the asymptotic convergence of the estimation error of the developed observer. This derived LMI condition is deduced by incorporating the reformulated Lipschitz property, a new variant of Young inequality, and the well-known linear parameter varying (LPV) approach. It is less conservative than the existing ones in the literature, and its effectiveness is compared with the other proposed approaches. Further, the proposed observer methodology is extended for the disturbance-affected nonlinear system for the purpose of the reconstruction of states and faults with optimal noise attenuation. Later on, both designed approaches are validated through an application of a single-link robotic arm manipulator in MATLAB Simulink.
\end{abstract}


\maketitle

\section{Introduction}
\label{sec 1}

Over the last few decades, due to advancements in technology, the susceptibility of industrial systems towards 
various types of exogenous signals such as faults or attacks have increased. It leads to the degrading of the reliability and safety of these systems. Hence, the task of obtaining information regarding these faults is very crucial and it is commonly known as fault estimation. The topic of fault estimation has become a centre of attraction in the domain of control system engineering. It is widely employed in various applications, such as the satellite formation flight~\cite{2018_Fligt_AT18}, the smart grid~\cite{FAULT_est_grid}, the autonomous vehicle~\cite{fault_est_autonomous_vehicle}, and so on.

The topic of fault diagnosis is concerned with the detection of faults and the identification of their nature.  
Fault estimation is one of the important tasks in the area of fault diagnosis.
The purpose of the fault estimation methodologies is to reconstruct the fault signals that occur in the system. The authors of~\cite{gao2022fast} have classified fault diagnosis methods into the subsequent ways: 1) Analytical/model-based methods~\cite{zhang2013robust_model_based}; 2) Signal-based methods~\cite{chen2012fault_based_signal}; 3) Knowledge-based methods~\cite{li2020process_knowledge_based}; and 4) Hybrid methods~\cite{chen2023random_hybrid}. In this article, the authors are mainly focusing on model-based fault diagnosis methods. In this approach, an unknown input observer is designed to reconstruct the states of the system along with the estimation of fault signals. In~\cite{witczak2016lmi_acutator_only}, authors had developed an LMI-based unknown input observer to estimate the actuator faults along the states of the systems, whereas, the authors of~\cite{zhang2008adaptive} have deployed an adaptive observer-based approach for the estimation of sensor faults. Additionally, abundant research has been carried out in the domain of simultaneous estimation of actuator and sensor faults. One can refer to~\cite{2018_Fligt_AT18},~\cite{gao2022fast} and~\cite{makni2017robust} for more details about this. 

Recently, several methodologies based on observer have been developed in the literature, for example, adaptive observers~\cite{zhang2008adaptive,jiang2006fault_adaptive}, unknown input observers (UIOs)~\cite{2018_Fligt_AT18}, neural network-based  observers~\cite{talebi2012_neural_network}, sliding mode observers~\cite{alwi2014robust_SMO} and descriptor
system observers~\cite{gao2007actuator_descriptor}.
The adaptive observers-based scheme relies on the utilization of adaptive law, which facilitates the estimation of actuator or sensor faults. Recently, the authors of~\cite{zhang2008adaptive} proposed a fast adaptive control law for estimation purposes, which provides a fast convergence of estimated faults. Though it provides high estimation accuracy, it can not be utilised for the estimation of both sensors and actuators faults at the same time.
The UIO-based methods are simple in terms of designing, and aiding in the states under the influence of unknown inputs such as the actuator faults. However, one of the important criteria for the development of these methods is the observer-matching condition, which is very restrictive in various practical applications. In the sliding mode observer approach, the estimation accuracy of estimated faults and states is very high as compared to the other aforementioned techniques, the occurrence of the chattering phenomenon during the estimation is one of the limitations of this method. The methodology of descriptor system-based observer is widely used in the simultaneous estimation of the states and faults. However, it can be used in the case of multiplicative faults. 
Hence, there is a scope for improvement in fault estimation strategies. 

The motivation behind this article is to design a novel observer-based method that aids in the estimation of states as well as both actuator and sensor faults (irrespective of whether it is constant or time-varying in nature). The authors utilized the earlier-stated approaches for the robust fault estimation of the nonlinear systems. The authors have incorporated the fast adaptive unknown input observer (FAUIO) approach proposed in~\cite{gao2022fast} with the fast adaptive control law of~\cite{zhang2008adaptive}. Further, the authors established a novel LMI condition for the computation of the parameters of FAUIO which is based on the method of~\cite{mohite_hal-04171967}.
The proposed approach is illustrated as follows:
\begin{enumerate}
    \item A descriptor system is formulated by augmenting the states of the system and the sensor fault vectors of a nonlinear system. Further, the FAUIO-based estimation strategy is deployed on the descriptor system for the simultaneous estimation of states, sensor faults, and actuator faults.
    \item An $\mathcal{H}_\infty$ criterion-based LMI condition is developed to determine the FAUIO parameters such that the proposed observer provides the estimated states and faults. 
    \item Further, the proposed methodology is extended for the state and fault estimation of disturbance-affected nonlinear systems under the influence of sensor and actuator faults with optimal noise attenuation.
\end{enumerate}

The main contribution of this article is as follows:
\begin{enumerate}[I)]
    \item Integration of FAUIO~\cite{gao2022fast} with the fast adaptive control law: In FAUIO, a learning rate factor is integrated to enhance the convergence rate of the estimation of fault vectors.
    \item Reducing the LMI constraints:
    In literature, various LMI-based fault estimation methodologies are based on several constraints such as the definiteness of several matrices, and the existence of invertibility of matrices. One can refer to the~\cite[LMI (5)-(6)]{jiang2006fault_adaptive},~\cite[LMI (32)-(33)]{zhang2008adaptive},~\cite[LMI (66)–(67) and (34)]{gao2022fast}. However, the authors proposed a single LMI condition in this article due to the judicious utilization of the new variant of Young inequality, which reduces the computation complexity and simplifies the design process.
    \item One of the major contributions of this method is the handling of the nonlinearities of the systems. The authors of~\cite{witczak2016lmi_acutator_only, gao2022fast, li2016state} employed a standard Lipschitz property to solve the problem of nonlinearity. However, in this paper, the authors have utilized a reformulated Lipschitz property~\cite{zemouche2013lmi} along with the well-known Linear parameter varying (LPV) approach to derive a novel matrix-multiplier-based LMI condition. 
\end{enumerate}

The remainder of this article is structured in the following manner:
The prerequisites required for the observer design procedure are described in Section~\ref{sec 2 Preliminaries}. The development of the observer and the LMI condition for the nonlinear systems with linear outputs under the presence of actuator and sensor faults is showcased in Section~\ref{sec 3 NDLO}. Further, Section~\ref{sec 4 NLDO with disturbance} is dedicated to the extension of the proposed approach for the nonlinear systems with linear outputs under the influence of actuators, sensor faults, and external disturbances/noise. Through the utilization of numerical examples, the effectiveness of the developed approaches is emphasised in Section~\ref{sec num examples}. Finally, Section~\ref{sec conclusion} encompasses some concluding remarks and future prospects of the established methodology.

\textbf{Notations:}
Throughout this letter, the ensuing notations are deployed:
\begin{itemize}
\item We have utilized the subsequent  vector of the canonical basis of $\mathbb{R}^{s}$ in this article: $$e_s(i) = (\underbrace{0,\hdots,0,\overbrace{1}^{i^{\text{th}}}, 0, \hdots, 0}_{s \,\, \text{components}})^\top \in \mathbb{R}^{s}, s\geq 1.$$
\item The term $x_0$ is used to denote the initial values of $x(t)$ at $t=0$.
    \item $||x||$ and $||x||_{\mathcal{L}_2}$ represent the euclidean norm and the $\mathcal{L}_2$ norm of a vector $x$, respectively. 
    \item An identity matrix and a null matrix are described by $\mathbb{I}$ and $\mathbb{O}$, respectively. 
    \item The symbol $(\star)$ is used to showcase the repeated blocks within a symmetric matrix. 
    \item 
    $A^\dagger$ and $A^\top$ depict the Moore–Penrose inverse and the transpose of the matrix $A$, respectively.
    \item $A \in \mathbf{S}^{n}$ infers that $A $ is a symmetric matrix of dimension ${n \times n}$.
    \item  For a matrix $A \in \mathbb{R}^{n \times n}$, $A > 0$ ($A < 0$) implies that $A$ is a positive definite matrix (a negative definite matrix). Similarly, a positive semi-definite matrix (a negative semi-definite matrix) is symbolised by $A \geq 0$ ($A \leq 0$). 
    \item $A = \text{block-diag}(
A_1, \hdots,A_n)$ infers  a block-diagonal matrix having elements $A_1, \hdots, A_n$ in the diagonal.   
\end{itemize}
\section{Preliminaries}\label{sec 2 Preliminaries}
In this section, the mathematical tools and some background results are explained.

\begin{lem}[\cite{zemouche2013lmi}]\label{Lem1}
If there exists a global Lipschitz nonlinear function $h: \mathbb{R}^n \to \mathbb{R}^n$, then,
\begin{itemize}
    \item there exist functions $h_{ij} : \mathbb{R}^n \times \mathbb{R}^n \to \mathbb{R}$, and constants $\underbar{h}_{ij}$ and $\bar{h}_{ij}$, $\forall~i,j = 1,\hdots,n$, such that $\forall\,X,\,Y \in \mathbb{R}^n$,
    \begin{equation}
        h(X)-h(Y) = \sum^{n}_{i=1}
       \sum^{n}_{j=1}h_{ij} \mathcal{H}_{ij} (X-Y)\label{L 2.2},
    \end{equation}
    where $\mathcal{H}_{ij} = e_n(i)e^\top_n(j)$, and $h_{ij}  \triangleq h_{ij}(X^{Y_{j-1}},X^{Y_{j}})$. The functions $h_{ij}(.)$ are globally bounded from above and
    below as follows:
    \begin{equation}
       \underbar{h}_{ij}\leq
      h_{ij}  \leq
       \bar{h}_{ij}.\label{L 2.3}
    \end{equation}
\end{itemize}
\end{lem}

\medskip
\begin{lem}
\label{Lem2}
Let $X, Y \in \mathbb{R}^n$ be two vectors and a matrix $Z>0 \in \mathbf{S}^{n}$. Then, the following inequality holds:
\begin{equation}\label{L 3.1}
    X^\top Y+Y^\top X \leq X^\top Z^{-1}X + Y^\top ZY.
\end{equation}
The inequality~\eqref{L 3.1} is known as the standard Young inequality.
In addition to this, the authors of~\cite{zemouche2017circle} proposed the ensuing new variant of this inequality:
\begin{equation}\label{L 3.1.1}
    X^\top Y+Y^\top X \leq \frac{1}{2}(X+ZY)^\top Z^{-1}(X+ZY).
\end{equation}

\end{lem}
\section{A nonlinear FAUIO design}\label{sec 3 NDLO}
\subsection{Delineating the problem statement}
Consider the following set of equations which represents a class of nonlinear systems with linear outputs under the presence of actuator and sensor faults (or cyber-attacks on sensor and actuator or a combination of both):
\begin{align}\label{sec 3 eq 1}
\begin{split}
     \dot{x}  &= A x +B u
     +G g(x)+E_f f_a,\\
     y        &= C x+D_f f_s,     
\end{split}
\end{align}
where $x \in \mathbb{R}^n$ and 
$y \in \mathbb{R}^p$ are the system states, and the system's outputs respectively. $u\in \mathbb{R}^s$
is the system input.
$A \in \mathbb{R}^{n \times n}$, $C \in \mathbb{R}^{p \times n}$, $G \in \mathbb{R}^{m \times n}$,
and  $B\in \mathbb{R}^{n \times s} $ are constant matrices. The $f_a \in \mathbb{R}^{a_1}$ and $f_s \in \mathbb{R}^{a_2}$ are the actuator faults (or cyber-attacks), and sensor faults (or cyber-attacks), respectively. $E_f  \in \mathbb{R}^{n \times a_1}$ and $D_f  \in \mathbb{R}^{p \times a_2}$  are constant matrices.  The function $f: \mathbb{R}^n \to \mathbb{R}^m$ is assumed to be globally Lipschitz, and it is represented under the ensuing form:
\begin{equation}\label{sec 3 eq 2}
g(x)=\begin{bmatrix}
 g_1(H_1 x)\\\vdots\\g_i(\underbrace{H_i x}_{\bar{\nu}_i})\\\vdots\\g_m(H_m x)
\end{bmatrix},
\end{equation}
where $H_i\in\mathbb{R}^{\bar{n} \times n}$ for all $i\in \{1,\hdots,m\}$.

Let us consider that system~\eqref{sec 3 eq 1} hold the subsequent assumptions:
\begin{assum}\label{Asumption 1}
 The system~\eqref{sec 3 eq 1} is observable, i.e., the pair $(A,~C)$ is detectable.
\end{assum}
\begin{assum}\label{Asumption 2}
Without loss of generality, matrices $E_f$ and $D_f$ are full column rank~\cite{gao2022fast}. 
\end{assum}
\begin{assum}\label{Asumption 3}
The first-order derivative of actuator fault signals $f_a$ is presumed to be $\mathcal{L}_2$ bounded. 
\end{assum}

For the purpose of simultaneous  estimation of the states and attacks/faults, the subsequent descriptor system is designed by considering the augmented state vector as $\zeta=\begin{bmatrix}
x\\f_s
\end{bmatrix} \in \mathbb{R}^{n_{\text{new}}}$:
\begin{align}\label{sec 3 eq 3}
    \begin{split}
\mathcal{T} \dot{\zeta} &=
A_\zeta \zeta + B u +
G g(\mathcal{T}\zeta)+ E_f f_a,\\
y&=\bar{C} \zeta ,       
    \end{split}
\end{align}
where $\mathcal{T}=\begin{bmatrix}
\mathbb{I} & \mathbb{O}
\end{bmatrix}$, $A_\zeta = \begin{bmatrix}
A & \mathbb{O}
\end{bmatrix}$, $\bar{C} = \begin{bmatrix}
C & D_f
\end{bmatrix}$ and $n_{\text{new}}=n+a_2$. In addition to this,$g(\mathcal{T}\zeta)$ is expressed as:
\begin{equation}\label{sec 3 eq 4}
g(\mathcal{T}\zeta)=\begin{bmatrix}
 g_1(H_1 \mathcal{T}\zeta)\\\vdots\\g_i(\underbrace{H_i \mathcal{T}\zeta}_{\nu_i})\\\vdots\\g_m(H_m \mathcal{T}\zeta)
\end{bmatrix}.
\end{equation}


The observer structure for the system~\eqref{sec 3 eq 3} is given by
\begin{equation}\label{sec 3 eq 5 obs}
    \begin{split}
\dot{\eta} &= N \eta + J y + L_1 B u + L_1 G g(\mathcal{T} \hat{\zeta})+L_1 E_f \hat{f}_a,\\
\hat{\zeta}&= \eta +F y,\\
\dot{\hat{f}}_a&=\beta L_2 (\tilde{y}+\dot{\tilde{y}}),\\
     \tilde{y}&=y-\bar{C}\hat{\zeta},
    \end{split}
\end{equation}
where $\hat{\zeta}$ and $\hat{f}_a$ are estimated augmented state $\zeta$ and actuator fault $f_a$, respectively. $N\in \mathbb{R}^{n_{\text{new}} \times n_{\text{new}}}$, $J \in \mathbb{R}^{n_{\text{new}}\times (p)}$, $L_1 \in \mathbb{R}^{n_{\text{new}}\times (n)}$, $F \in \mathbb{R}^{n_{\text{new}}\times (p)}$ and $L_2 \in \mathbb{R}^{(a_1)\times p}$ are unknown constant matrices, which are also described as the observer parameters.  Additionally, a positive scalar $\beta$ is called a learning rate~\cite{zhang2008adaptive}.

We consider the subsequent definitions:
\begin{enumerate}[i)]
    \item Augmented state estimation error: $\tilde{\zeta}=\zeta-\hat{\zeta}$,
    \item Estimation error of actuator fault: $\tilde{f}_a=f_a-\hat{f}_a$.
\end{enumerate}
From these definitions, one can easily obtain
\begin{equation}\nonumber
    \tilde{\zeta}=\zeta-\eta-F\bar{C}\zeta=(\mathbb{I}-F\bar{C})\zeta -\eta.
\end{equation}
Now, let us presume the observer parameters $L_1$ and $F$ of~\eqref{sec 3 eq 5 obs} satisfy the following condition:
\begin{equation}\label{sec 3 eq 6 UIO con 1}
L_1\mathcal{T}+F\bar{C}=\mathbb{I}_{n_{\text{new}}}.
\end{equation}
It yields:
\begin{equation}\label{sec 3 eq 7}
    \tilde{\zeta}=L_1 \mathcal{T}\zeta -\eta.
\end{equation}
Further, through the utilisation of~\eqref{sec 3 eq 3} and~\eqref{sec 3 eq 5 obs}, we get
\begin{equation}\label{sec 3 eq 6 er zeta}
\begin{split}
  \dot{\tilde{\zeta}}&=N \tilde{\zeta} +(L_1 A_\zeta-NL_1 \mathcal{T}-J\bar{C})\zeta\\&+L_1 G\bigg( g(\mathcal{T}\zeta)- g(\mathcal{T} \hat{\zeta})\bigg)+
L_1 E_f\tilde{f}_a  
\end{split}  
\end{equation}
For lucidity of presentation, let us consider
\begin{equation}\label{sec 3 eq 6 UIO con 2 K}
    \mathcal{K}=J-NF,
\end{equation}
and
\begin{equation}\label{sec 3 eq 6 UIO con 2}
    L_1 A_\zeta-NL_1 \mathcal{T}-J\bar{C}=0.
\end{equation}
From~\eqref{sec 3 eq 6 UIO con 2 K} and~\eqref{sec 3 eq 6 UIO con 2}, 
\begin{equation}\label{sec 3 eq 6 UIO con 3}
    N=L_1 A_\zeta-\mathcal{K}\bar{C}.
\end{equation}
In addition to this, Lemma~\ref{Lem1} is deployed on the term $g(\mathcal{T}\zeta)- g(\mathcal{T} \hat{\zeta})$, and we achieve:
\begin{align}\label{sec 3 eq 7 g tilde}
    \begin{split}
g(\mathcal{T}\zeta)- g(\mathcal{T} \hat{\zeta})&=\sum_{i,j=1}^{m,\bar{n}} g_{ij}\mathcal{H}_{ij} H_i\mathcal{T} \Tilde{\zeta},
    \end{split}
\end{align}
where $g_{ij} \triangleq g_{ij}(\nu_i^{\hat{\nu}_{i,j-1}},\nu_i^{\hat{\nu}_{i,j}})$ and
$\mathcal{H}_{i,j} = e_{m}(i)e^\top_{\bar{n}}(j)$. The functions $g_{ij}$ fulfils:
\begin{equation}\nonumber
    {g}_{a_{ij}} \leq  g_{ij}\leq {g}_{b_{ij}},
\end{equation}
where ${g}_{a_{ij}}$ and ${g}_{b_{ij}}$ are known constants. Without loss of generality, $f_{a_{ij}}$ is considered as $f_{a_{ij}}=0$. Thus, it leads to:
\begin{equation}\label{sec 3 eq 7.1}
    0 \leq  g_{ij}\leq {g}_{b_{ij}},
\end{equation}
One can refer to~\cite{zemouche2017circle}
for more details about this.
\\By using~\eqref{sec 3 eq 6 UIO con 2},~\eqref{sec 3 eq 6 UIO con 3} and~\eqref{sec 3 eq 7 g tilde}, the error dynamic~\eqref{sec 3 eq 6 er zeta} is modified to
\begin{equation}\label{sec 3 eq 8 er zeta final}
\begin{split}
    \dot{\tilde{\zeta}}&=(L_1 A_\zeta-\mathcal{K}\bar{C}) \tilde{\zeta} +\bigg(\sum_{i,j=1}^{m,\bar{n}} g_{ij}L_1 G\mathcal{H}_{ij} H_i\mathcal{T} \bigg)\Tilde{\zeta}\\& +
L_1 E_f\tilde{f}_a.
\end{split}
\end{equation}
Further, the dynamic of fault estimation error ($\tilde{f}_{a}=f_a-\hat{f}_a$)  is illustrated as:
\begin{equation}\nonumber
        \dot{\tilde{f}}_{a}=\dot{f}_a-\dot{\hat{f}}_a=-\beta L_2 \bar{C} \Tilde{\zeta} -\beta L_2 \bar{C} \dot{\Tilde{\zeta}}+\dot{f}_a.
\end{equation}
Hence,
\begin{equation}\label{sec 3 eq 9 er fa final}
        \dot{\tilde{f}}_{a}=-\beta L_2 \bar{C} \Tilde{\zeta} -\beta L_2 \bar{C} \dot{\Tilde{\zeta}}+\dot{f}_a.
\end{equation}
In order to avoid cumbersome equations, let us introduce:
$e=\begin{bmatrix}
    \Tilde{\zeta}\\\tilde{f}_{a}
\end{bmatrix},\,\mathcal{T}_e=\begin{bmatrix}
\mathbb{I}_{n_{\text{new}}}&\mathbb{O}\\\beta L_2 \bar{C} & \mathbb{I}_{a_1}
\end{bmatrix},\,\mathcal{A}_e= \begin{bmatrix}
    L_1 A_\zeta-\mathcal{K}\Bar{C} & L_1 E_f\\-\beta L_2 \bar{C} & \mathbb{O}
    \end{bmatrix}.$
Through the utilisation of~\eqref{sec 3 eq 8 er zeta final},~\eqref{sec 3 eq 9 er fa final} and the aforementioned notations, we deduce:
\begin{equation}\small\label{sec 3 eq 10}
    \mathcal{T}_e\dot{e}=\mathcal{A}_e e+\Bigg[\sum_{i,j=1}^{m,\bar{n}}  \begin{bmatrix}
    L_1 G\mathcal{H}_{ij}\\\mathbb{O}
    \end{bmatrix}\begin{bmatrix}
    g_{ij}H_i\mathcal{T}&\mathbb{O}
    \end{bmatrix}\Bigg]e+\begin{bmatrix}
        \mathbb{O}\\ \mathbb{I}_{a_1}
    \end{bmatrix}\dot{f}_a.
\end{equation}
It leads to
\begin{equation}\label{sec 3 eq 11}
\begin{split}
    \dot{e}&=\mathcal{T}_e^{-1}\mathcal{A}_e e+\Bigg[\sum_{i,j=1}^{m,\bar{n}} \mathcal{T}_e^{-1} \begin{bmatrix}
    L_1 G\mathcal{H}_{ij}\\\mathbb{O}
    \end{bmatrix}\begin{bmatrix}
    g_{ij}H_i\mathcal{T}&\mathbb{O}
\end{bmatrix}\Bigg]e\\&+\mathcal{T}_e^{-1}\begin{bmatrix}
        \mathbb{O}\\ \mathbb{I}_{a_1}
    \end{bmatrix}\dot{f}_a,
\end{split}
\end{equation}
where $\mathcal{T}_e^{-1}=\begin{bmatrix}
    \mathbb{I}_{n_{\text{new}}}&\mathbb{O}\\-\beta L_2 \bar{C} & \mathbb{I}_{a_1}
\end{bmatrix}$. 
Let us consider the following matrices:
\begin{equation}\scriptsize\label{sec 3 eq 12}
\tilde{\mathcal{A}}_e=\mathcal{T}_e^{-1}\mathcal{A}_e=\begin{bmatrix}
 L_1 A_\zeta-\mathcal{K}\Bar{C} & L_1 E_f \\
 -\beta L_2 \bar{C} L_1 A_\zeta+\beta L_2 \bar{C} \mathcal{K} \bar{C}-\beta L_2 \bar{C} & -\beta L_2 \bar{C} L_1 E_f
\end{bmatrix},
\end{equation}
and
\begin{equation}\label{sec 3 eq 13}
   \mathcal{T}_e^{-1} \begin{bmatrix}
    L_1 G\mathcal{H}_{ij}\\\mathbb{O}
    \end{bmatrix}=\begin{bmatrix}
   L_1 G\mathcal{H}_{ij}\\
    -\beta L_2 \bar{C} L_1 G \mathcal{H}_{ij}
    \end{bmatrix}.
\end{equation}
Thus, we get:
\begin{equation}\label{sec 3 eq 14}
\begin{split}
\dot{e}&=\tilde{\mathcal{A}}_e e+\Bigg[\sum_{i,j=1}^{m,\bar{n}} \begin{bmatrix}
   L_1 G\mathcal{H}_{ij}\\
    -\beta L_2 \bar{C} L_1 G \mathcal{H}_{ij}
    \end{bmatrix}\begin{bmatrix}
    g_{ij}H_i\mathcal{T}&\mathbb{O}
\end{bmatrix}\Bigg]e\\&+\begin{bmatrix}
        \mathbb{O}\\ \mathbb{I}_{a_1}
    \end{bmatrix}\dot{f}_a.    
\end{split}
\end{equation}
According to~\cite{gao2022fast}, the necessary conditions for the existence of FAUIO~\eqref{sec 3 eq 5 obs} are:
\begin{enumerate}[I)]
    \item $(L_1A_{\eta},\Bar{C})$ must be detectable, i.e.,
\begin{equation}\label{sec 3 rank con 1}
        \text{Rank of}
        \begin{bmatrix}
        L_1A_\zeta-\lambda \mathbb{I}\\
        \Bar{C}
\end{bmatrix}=n_\text{new}\quad \forall \lambda \in \mathbb{C}^n.
    \end{equation}
    \item The matrix 
    $\begin{bmatrix}
    L_1A_\zeta & L_1E_f\\
    \Bar{C}&\mathbb{O}
    \end{bmatrix}$ is a full-column matrix. In other words, $(L_1A_{\zeta},L_1E_f,
    \Bar{C})$ does not have invariant zeros at the origin, that is,
    \begin{equation}\label{sec 3 rank con 2}
      \text{Rank of} \begin{bmatrix}
    L_1A_\eta & L_1E_f\\
\Bar{C}&\mathbb{O}\end{bmatrix}=\underbrace{n_\text{new}+a_1}_{n_{a_1}}
    \end{equation}
\end{enumerate}
Due to the fact that $\bar{C}$ is a full column matrix, conditions~\eqref{sec 3 rank con 1} and~\eqref{sec 3 rank con 2} are fulfilled.

The aim of this section is to determine the observer parameters so that the error dynamics~\eqref{sec 3 eq 14} satisfies the subsequent $\mathcal{H}_\infty$ criterion:
\begin{equation}\label{sec 3 eq 15 H infinity}
  ||e||_{\mathcal{L}^{n_{a_1}}_2}  \leq \sqrt{\nu ||e_0||^2+\mu  ||\dot{f}_a||^2_{\mathcal{L}^{a_1}_2}},
\end{equation}
where $\mu,\nu > 0$. The term $\sqrt{\mu}$ denotes the disturbance attenuation level.

\subsection{LMI formulation}
Let us consider the ensuing quadratic Lyapunov function for $\mathcal{H}_\infty$ stability analysis of~\eqref{sec 3 eq 14}:
\begin{equation}\label{sec 3 eq LY fun}
    V(e)= e^\top P e,
\end{equation}
where $P=\begin{bmatrix}
    P_1 & \mathbb{O}\\\mathbb{O}&\beta^{-1}P_2
    \end{bmatrix},\,\,
    P_1^\top=P_1>0\in \mathbb{R}^{n_{\text{new}}\times n_{\text{new}}},P_2^\top=P_2>0\in \mathbb{R}^{a_1 \times a_1}$.
\\The derivative of $V(e)$ is computed along the trajectories of~\eqref{sec 3 eq 14}, and given by
\begin{equation}\scriptsize\label{sec 3 eq 15 vdot}
    \begin{split}
    \dot{V}(e)&=   e^\top \Bigg[\Tilde{\mathcal{A}}_e^\top P+P\Tilde{\mathcal{A}}_e\Bigg]e
    +e^\top \Bigg[\sum_{i,j=1}^{m,\bar{n}}
P\begin{bmatrix}
   L_1 G\mathcal{H}_{ij}\\
    -\beta L_2 \bar{C} L_1 G \mathcal{H}_{ij}
    \end{bmatrix}\begin{bmatrix}
    g_{ij}H_i\mathcal{T}&\mathbb{O}
\end{bmatrix}\\&+\bigg(P\begin{bmatrix}
   L_1 G\mathcal{H}_{ij}\\
    -\beta L_2 \bar{C} L_1 G \mathcal{H}_{ij}
    \end{bmatrix}\begin{bmatrix}
    g_{ij}H_i\mathcal{T}&\mathbb{O}
\end{bmatrix}\bigg)^\top\Bigg]e\\&+e^\top P \begin{bmatrix}
        \mathbb{O}\\ \mathbb{I}_{a_1}
\end{bmatrix}\dot{f}_a+\dot{f}_a^\top \begin{bmatrix}
        \mathbb{O}\\ \mathbb{I}_{a_1}
    \end{bmatrix}^\top P e.
    \end{split}
\end{equation}
According to~\cite{zemouche2009unified}, the error dynamic~\eqref{sec 3 eq 14} satisfy the $\mathcal{H}_\infty$ criterion~\eqref{sec 3 eq 15 H infinity} if it admits a Lyapunov function~\eqref{sec 3 eq LY fun} that fulfills 
\begin{equation}\label{sec 4 eq 16}
  \mathcal{W} \triangleq  \dot{V}(e)+||e||^2-\mu ||\dot{f}_a||^2 \leq 0.
\end{equation}
Through the utilisation of~\eqref{sec 3 eq 15 vdot} and~\eqref{sec 4 eq 16}, $\mathcal{W}\leq 0$ if
\begin{equation}\label{sec 4 eq 17 w con}
    \begin{bmatrix}
\mathcal{A}_n &\Sigma_{q} \\
    \star &-\mu \mathbb{I}
\end{bmatrix}\leq 0,
\end{equation}
where 
\begin{equation}
\begin{split}\label{sec 4 eq 17 w con term 1}
    \mathcal{A}_n &=\Tilde{\mathcal{A}}_e^\top P+P\Tilde{\mathcal{A}}_e+\mathbb{I}_{n_{a_1}}\\&+\Bigg[\sum_{i,j=1}^{m,\bar{n}}
P\begin{bmatrix}
   L_1 G\mathcal{H}_{ij}\\
    -\beta L_2 \bar{C} L_1 G \mathcal{H}_{ij}
    \end{bmatrix}\begin{bmatrix}
    g_{ij}H_i\mathcal{T}&\mathbb{O}
\end{bmatrix}\\&+\bigg(P\begin{bmatrix}
   L_1 G\mathcal{H}_{ij}\\
    -\beta L_2 \bar{C} L_1 G \mathcal{H}_{ij}
    \end{bmatrix}\begin{bmatrix}
    g_{ij}H_i\mathcal{T}&\mathbb{O}
\end{bmatrix}\bigg)^\top\Bigg],
\end{split}
\end{equation}
and
\begin{equation}\label{sec 4 eq 17 w con term 2}
    \Sigma_{q}=P \begin{bmatrix}
        \mathbb{O}\\ \mathbb{I}_{a_1}
\end{bmatrix}=\begin{bmatrix}
        \mathbb{O}\\ \beta^{-1}P_2
\end{bmatrix}.
\end{equation}
Further,
\begin{equation}
    \begin{split}\label{sec 4 eq 18 Ae1}
\Tilde{\mathcal{A}}_e^\top P+P \Tilde{\mathcal{A}}_e+\mathbb{I}&= \Sigma_{11}+\begin{bmatrix}
    \mathbb{O}&(R_2^\top \bar{C} \mathcal{K}\bar{C})^\top\\
    \star&\mathbb{O}
\end{bmatrix},
    \end{split}
\end{equation}
where
$R_1^\top=P_1 \mathcal{K}$; $R_2^\top=P_2 L_2$ and $\Sigma_{11}$ is illustrated in~\eqref{sec 4 eq 18 sigma11} 
\begin{table*}[!ht]
\begin{equation}\label{sec 4 eq 18 sigma11}
\Sigma_{11}=\begin{bmatrix}
 P_1 L_1 A_\zeta+(P_1 L_1 A_\zeta)^\top-R_1^\top \bar{C} -\bar{C}^\top R_1 & P_1 L_1 E_f- (R_2^\top \bar{C} L_1 A_\zeta)^\top- (R_2^\top \bar{C})^\top \\
 (P_1 L_1 E_f)^\top- R_2^\top \bar{C} L_1 A_\zeta- R_2^\top \bar{C} & -R_2^\top \bar{C} L_1 E_f-(R_2^\top\bar{C} L_1 E_f)^\top
\end{bmatrix}+\mathbb{I}, 
\end{equation}
\hrule
\end{table*}
\\Additionally,
\begin{equation}\label{sec 4 eq 18 Ae1 1}
    \begin{split}
 \begin{bmatrix}
    \mathbb{O}&(R_2^\top \bar{C} \mathcal{K}\bar{C})^\top\\
     \star&\mathbb{O}
\end{bmatrix}= \underbrace{\begin{bmatrix}
\mathbb{O}\\R_2^\top \bar{C}
\end{bmatrix}}_{\mathbb{M}^\top}\overbrace{P_1^{-1}
\begin{bmatrix}
R_1^\top \bar{C}&\mathbb{O}
\end{bmatrix}}^{\mathbb{N}_1}+\mathbb{N}_1^\top \mathbb{M} 
    \end{split}
\end{equation}
We derive the ensuing inequality by deploying~\eqref{L 3.1.1} on~\eqref{sec 4 eq 18 Ae1 1}:
\begin{equation}\label{sec 4 eq 18 Ae1 1f}
 \mathbb{N}_1^\top \mathbb{M} +
 \mathbb{M}^\top \mathbb{N}_1\leq (\mathbb{M}+ \epsilon \mathbb{N})^\top (-2\epsilon P_1)^{-1}(\mathbb{M}+ \epsilon \mathbb{N}),     
\end{equation}
where 
$\mathbb{N}=P_1 \mathbb{N}_1=\begin{bmatrix}
    R^\top \bar{C} & \mathbb{O}
\end{bmatrix}$
and $\epsilon$ is a positive scalar.

From~\eqref{sec 4 eq 18 sigma11} and~\eqref{sec 4 eq 18 Ae1 1f}, one can rewrite~\eqref{sec 4 eq 18 Ae1} as:
\begin{equation}\label{sec 4 eq 18 wcon Ae1}
 \Tilde{\mathcal{A}}_e^\top P+P \Tilde{\mathcal{A}}_e+\mathbb{I}\leq \Sigma_{11}+ (\mathbb{M}+ \epsilon \mathbb{N})^\top (-2\epsilon P_1)^{-1}(\mathbb{M}+ \epsilon \mathbb{N}).  
\end{equation}
Let us define
\begin{equation}\label{sec 4 eq 17 NL1 1}
\begin{split}
 \mathbf{NL}_1&= \Bigg[\sum_{i,j=1}^{m,\bar{n}}
P\begin{bmatrix}
   L_1 G\mathcal{H}_{ij}\\
    -\beta L_2 \bar{C} L_1 G \mathcal{H}_{ij}
    \end{bmatrix}\begin{bmatrix}
    g_{ij}H_i\mathcal{T}&\mathbb{O}
\end{bmatrix}\\&+\bigg(P\begin{bmatrix}
   L_1 G\mathcal{H}_{ij}\\
    -\beta L_2 \bar{C} L_1 G \mathcal{H}_{ij}
    \end{bmatrix}\begin{bmatrix}
    g_{ij}H_i\mathcal{T}&\mathbb{O}
\end{bmatrix}\bigg)^\top\Bigg]   
\end{split}
\end{equation}
Furthermore,
\begin{equation}\label{sec 4 eq 17 NL1}
    \mathbf{NL}_1=\sum_{i,j=1}^{m,\bar{n}}
\underbrace{
 \begin{bmatrix}
   P_1 L_1 G\mathcal{H}_{ij}\\
    - R_2^\top \bar{C} L_1 G \mathcal{H}_{ij}
\end{bmatrix}}_{\mathbb{X}_{ij}^\top} \overbrace{g_{ij}\underbrace{
  \begin{bmatrix}
    H_i\mathcal{T}&\mathbb{O}
\end{bmatrix}}_{\mathbb{H}_i}}^{\mathbb{Y}_{ij}}+\mathbb{Y}_{ij}^\top \mathbb{X}_{ij}
\end{equation}
Now, the subsequent notations are introduced for the lucidity of the presentation:
\begin{equation}\small\label{sec 4 eq 18 X}
    \mathbb{X}=\begin{bmatrix}
\mathbb{X}_{11}^\top&\mathbb{X}_{12}^\top&\hdots&\mathbb{X}_{1\bar{n}}^\top&\hdots&\mathbb{X}_{m1}^\top&\hdots&\mathbb{X}_{m\bar{n}}^\top
\end{bmatrix}^\top,
\end{equation}
\begin{equation}\scriptsize\label{sec 4 eq 18 Y}
    \mathbb{Y}=\begin{bmatrix}
\mathbb{Y}_{11}^\top&\mathbb{Y}_{12}^\top&\hdots&\mathbb{Y}_{1\bar{n}}^\top&\hdots&\mathbb{Y}_{m1}^\top&\hdots&\mathbb{Y}_{m\bar{n}}^\top
    \end{bmatrix}^\top=\mathbb{H}\Phi,
\end{equation}
where,
\begin{equation}\small\label{sec 4 eq 18 H}
\mathbb{H}=    \text{block-diag}\bigg(\mathbb{H}_1,\mathbb{H}_1,\hdots,\mathbb{H}_1,\hdots,
\mathbb{H}_m,\hdots,\mathbb{H}_m\bigg),
\end{equation}
\begin{equation}\scriptsize\label{sec 4 eq 18 phi}
 \Phi=   \begin{bmatrix}
g_{11}\mathbf{I}_{\bar{n}}&g_{12}\mathbf{I}_{\bar{n}}& \hdots&
      g_{1\bar{n}}\mathbf{I}_{\bar{n}}& \hdots & 
      g_{m1}\mathbf{I}_{\bar{n}}& \hdots & g_{m\bar{n}}\mathbf{I}_{\bar{n}}
    \end{bmatrix}.
\end{equation}
By using all these aforementioned notations, $\mathbf{NL}_1$ is rewritten as:
\begin{equation}\label{sec 3 eq 19 NL1}
    \mathbf{NL}_1=\mathbb{X}^\top (\mathbb{H} \Phi)+(\mathbb{H} \Phi)^\top \mathbb{X}.
\end{equation}
\begin{table*}[!h]
\begin{equation}\scriptsize
\label{sec 4 EQN 11 Z}
\mathbb{Z} =   \left[
    \begin{array}{c c c c c c c c c c c c c}
    Z_{11} & Z_{a^1_{12}} & \hdots&Z_{a^1_{1\bar{n}}} & \textcolor{red}{Z_{b^{11}_{21}}} & \textcolor{red}{Z_{{b}^{11}_{22}}}  &\textcolor{red}{\hdots} &\textcolor{red}{Z_{{b}^{11}_{2\bar{n}}}}   &\textcolor{green}{\hdots}& \textcolor{mypink}{Z_{b^{11}_{m1}}}  & \textcolor{mypink}{Z_{{b}^{11}_{m2}}} &\textcolor{mypink}{\hdots} & \textcolor{mypink}{Z_{{b}^{11}_{m\bar{n}}}}
    \\
    Z_{a^1_{12}} & Z_{12} & \hdots &Z_{a^2_{1\bar{n}}}&\textcolor{red}{Z_{{b}^{12}_{21}}}& \textcolor{red}{Z_{{b}^{12}_{22}}} &\textcolor{red}{\hdots} & \textcolor{red}{Z_{{b}^{12}_{2\bar{n}}}} &\textcolor{green}{\hdots} &
    \textcolor{mypink}{Z_{{b}^{12}_{m1}}}& \textcolor{mypink}{Z_{{b}^{12}_{m2}}}&\textcolor{mypink}{\hdots} & \textcolor{mypink}{Z_{{b}^{12}_{m\bar{n}}}}\\
\vdots&\vdots&\ddots&\vdots&\textcolor{red}{\vdots}&\textcolor{red}{\vdots}&\textcolor{red}{\ddots}&\textcolor{red}{\vdots}&\textcolor{green}{\hdots} &\textcolor{mypink}{\vdots}&\textcolor{mypink}{\vdots}&\textcolor{mypink}{\ddots}&\textcolor{mypink}{\vdots}\\
Z_{a^1_{1\bar{n}}} & Z_{a^2_{1\bar{n}}} & \hdots & Z_{1\bar{n}}&\textcolor{red}{Z_{{b}^{1\bar{n}}_{21}}}& \textcolor{red}{Z_{{b}^{1\bar{n}}_{22}}}
&\textcolor{red}{\hdots} & \textcolor{red}{Z_{{b}^{1\bar{n}}_{2\bar{n}}}}&\textcolor{green}{\hdots} &
\textcolor{mypink}{Z_{{b}^{1\bar{n}}_{m1}}}& \textcolor{mypink}{Z_{{b}^{1\bar{n}}_{m2}}}
&\textcolor{mypink}{\hdots} & \textcolor{mypink}{Z_{{b}^{1\bar{n}}_{m\bar{n}}}}
\\
\textcolor{red}{Z_{b^{11}_{21}}} & \textcolor{red}{Z_{b^{12}_{21}}}  &\textcolor{red}{\hdots} &\textcolor{red}{Z_{{b}^{1\bar{n}}_{21}}}& Z_{21} &   Z_{a^1_{22} }& \hdots& Z_{a^1_{2\bar{n}}}&\textcolor{green}{\hdots}& \textcolor{mypink1}{Z_{{b}^{21}_{m1}}} & \textcolor{mypink1}{Z_{{b}^{21}_{m2}}}
&\textcolor{mypink1}{\hdots} & \textcolor{mypink1}{Z_{{b}^{21}_{m\bar{n}}}}\\
\textcolor{red}{Z_{b^{11}_{22}}} & \textcolor{red}{Z_{b^{12}_{22}}}  &\textcolor{red}{\hdots} &\textcolor{red}{Z_{{b}^{1\bar{n}}_{22}}}& Z_{a^1_{22}}&  Z_{22} &\hdots &Z_{a^2_{2\bar{n}}} &\textcolor{green}{\hdots} &
\textcolor{mypink1}{Z_{{b}^{22}_{m1}}} & \textcolor{mypink1}{Z_{{b}^{22}_{m2}}}
&\textcolor{mypink1}{\hdots} & \textcolor{mypink1}{Z_{{b}^{22}_{m\bar{n}}}}  \\
\textcolor{red}{\vdots}&
\textcolor{red}{\vdots}&
\textcolor{red}{\ddots}&
\textcolor{red}{\vdots}&\vdots&\vdots&\ddots&\vdots&\textcolor{green}{\hdots}&\textcolor{mypink1}{\vdots}&
\textcolor{mypink1}{\vdots}&
\textcolor{mypink1}{\ddots}&
\textcolor{mypink1}{\vdots}\\
\textcolor{red}{Z_{b^{11}_{2\bar{n}}}} & \textcolor{red}{Z_{b^{12}_{2\bar{n}}}}  &\textcolor{red}{\hdots} &\textcolor{red}{Z_{{b}^{1\bar{n}}_{2\bar{n}}}}& 
 Z_{a^1_{2\bar{n}}}&
 Z_{a^2_{2\bar{n}}}&\hdots&  Z_{2\bar{n}}&\textcolor{green}{\hdots}&\textcolor{mypink1}{Z_{{b}^{2\bar{n}}_{m1}}} & \textcolor{mypink1}{Z_{{b}^{2\bar{n}}_{m2}}}
&\textcolor{mypink1}{\hdots} & \textcolor{mypink1}{Z_{{b}^{2\bar{n}}_{m\bar{n}}}}
\\
\textcolor{green}{\vdots}&
\textcolor{green}{\vdots}&
\textcolor{green}{\vdots}&
\textcolor{green}{\vdots}&
\textcolor{green}{\vdots}&
\textcolor{green}{\vdots}&
\textcolor{green}{\vdots}&
\textcolor{green}{\vdots}&\ddots&
\textcolor{green}{\vdots}&
\textcolor{green}{\vdots}&
\textcolor{green}{\vdots}&
\textcolor{green}{\vdots}\\
\textcolor{mypink}{Z_{b^{11}_{m1}}} & \textcolor{mypink}{Z_{b^{12}_{m1}}}  &\textcolor{mypink}{\hdots} &\textcolor{mypink}{Z_{{b}^{1\bar{n}}_{m1}}}&\textcolor{mypink1}{Z_{b^{21}_{m1}}} & \textcolor{mypink1}{Z_{b^{22}_{m1}}}  &\textcolor{mypink1}{\hdots} &\textcolor{mypink1}{Z_{{b}^{2\bar{n}}_{m1}}}&\textcolor{green}{\hdots} & Z_{m1} &  Z_{a^1_{m2}} &\hdots &   Z_{a^1_{m\bar{n}}}\\
\textcolor{mypink}{Z_{b^{11}_{m2}}} & \textcolor{mypink}{Z_{b^{12}_{m2}}}  &\textcolor{mypink}{\hdots} &\textcolor{mypink}{Z_{{b}^{1\bar{n}}_{m2}}}&\textcolor{mypink1}{Z_{b^{21}_{m2}}} & \textcolor{mypink1}{Z_{b^{22}_{m2}}}  &\textcolor{mypink1}{\hdots} &\textcolor{mypink1}{Z_{{b}^{2\bar{n}}_{m2}}}&\textcolor{green}{\hdots} &
 Z_{a^1_{m2}}&  Z_{m2} &\hdots &  Z_{a^2_{m\bar{n}}}\\
\textcolor{mypink}{\vdots}&
\textcolor{mypink}{\vdots}&
\textcolor{mypink}{\ddots}&
\textcolor{mypink}{\vdots}&\textcolor{mypink1}{\vdots}&
\textcolor{mypink1}{\vdots}&
\textcolor{mypink1}{\ddots}&
\textcolor{mypink1}{\vdots}&\textcolor{green}{\hdots}&\vdots&\vdots&\ddots&\vdots\\
\textcolor{mypink}{Z_{b^{11}_{m\bar{n}}}} & \textcolor{mypink}{Z_{b^{12}_{m\bar{n}}}}  &\textcolor{mypink}{\hdots} &\textcolor{mypink}{Z_{{b}^{1\bar{n}}_{m\bar{n}}}}&\textcolor{mypink1}{Z_{b^{21}_{m\bar{n}}}} & \textcolor{mypink1}{Z_{b^{22}_{m\bar{n}}}}  &\textcolor{mypink1}{\hdots} &\textcolor{mypink1}{Z_{{b}^{2\bar{n}}_{m\bar{n}}}}&\textcolor{green}{\hdots}&
 Z_{a^1_{m\bar{n}}}& Z_{a^2_{m\bar{n}}}
&\hdots &  Z_{m\bar{n}}
      \end{array}
\right],
\end{equation}
 \begin{flushleft} 
where 
$Z_{ij}=Z_{ij}^\top>0\in \mathbb{R}^{\bar{n}\times \bar{n}},~Z_{a^k_{ij}}=Z_{a^k_{ij}}^\top\geq0\in \mathbb{R}^{\bar{n}\times \bar{n}}\forall i,k\in \{1,\hdots,m\}, \& j\in \{1,\hdots,\bar{n}\}$;$~Z_{b^{kj}_{ij}}=Z_{b^{kj}_{ij}}^\top\geq0\in \mathbb{R}^{\bar{n}\times \bar{n}},\forall i\in \{2,\hdots,m\},k\in \{1,\hdots,m-1\}, \& j\in \{1,\hdots,\bar{n}\}$ such that $\mathbb{Z}>0$.
\end{flushleft}
\hrule
\end{table*}
The following inequality is derived by deploying
a new variant of Young inequality~\eqref{L 3.1.1} on~\eqref{sec 3 eq 19 NL1}:
\begin{equation}\label{sec 3 eq 20 NL1}
    \mathbf{NL}_1\leq \frac{1}{2}\big(\mathbb{X}+\mathbb{Z}\mathbb{H} \Phi\big)^\top (\mathbb{Z})^{-1}\big(\mathbb{X}+\mathbb{Z}\mathbb{H} \Phi\big),
\end{equation}
where $\mathbb{Z}=\mathbb{Z}^\top>0$ and it is specified in~\eqref{sec 4 EQN 11 Z}.

Each element of $\Phi$ satisfies the boundary condition specified in~\eqref{sec 3 eq 7.1}. It infers that every component inside $\Phi$ belongs to a convex set $\mathcal{G}_m$ which is defined as follows:
\begin{equation}\small\nonumber
    \mathcal{G}_m\triangleq\big\{ \Phi : 0\leq g_{ij}\leq g_{b_{ij}}, \forall i\in \{1,\hdots,m\}\,\&\,j\in \{1,\hdots,\bar{n}\}\big\}.
\end{equation}
The set of vertices for $\mathcal{G}_m$ is described as:
\begin{equation}\scriptsize\label{sec 3 eq 22 GM} 
\mathcal{V}_{\phi}=\bigg\{
\{\mathcal{G}_{11},
\hdots,\mathcal{G}_{1\bar{n}},\hdots,\mathcal{G}_{m1},\hdots,\mathcal{G}_{m\bar{n}}\} : \mathcal{G}_{ij} \in [0,f_{b_{ij}}]
  \bigg\} .        
\end{equation}
It leads to:
\begin{equation}\label{sec 3 eq 20 NL1.1}
    \mathbf{NL}_1\leq \Bigg[\big(\mathbb{X}+\mathbb{Z}\mathbb{H} \Phi\big)^\top (2\mathbb{Z})^{-1}\big(\mathbb{X}+\mathbb{Z}\mathbb{H} \Phi\big)\Bigg]_{\Phi\in \mathcal{V}_{\phi}},
\end{equation}
Hence,~\eqref{sec 4 eq 17 w con term 1} is modified as:
\begin{equation}\label{sec 4 eq 17 w con term 1 fi}
\begin{split}
    \mathcal{A}_n&\leq \Sigma_{11}+ (\mathbb{M}+ \epsilon \mathbb{N})^\top (-2\epsilon P_1)^{-1}(\mathbb{M}+ \epsilon \mathbb{N})  \\&+\Bigg[\big(\mathbb{X}+\mathbb{Z}\mathbb{H} \Phi\big)^\top (2\mathbb{Z})^{-1}\big(\mathbb{X}+\mathbb{Z}\mathbb{H} \Phi\big)\Bigg]_{\Phi\in \mathcal{V}_{\phi}}
\end{split}
\end{equation}
Therefore,~\eqref{sec 4 eq 17 w con} holds if
the inequality~\eqref{sec 4 eq 17 w con main} is satisfied.
\begin{table*}[!ht]
 \begin{equation}\label{sec 4 eq 17 w con main}
    \begin{bmatrix}
   \Sigma_{11}+ (\mathbb{M}+ \epsilon \mathbb{N})^\top (-2\epsilon P_1)^{-1}(\mathbb{M}+ \epsilon \mathbb{N})  +\Bigg[\big(\mathbb{X}+\mathbb{Z}\mathbb{H} \Phi\big)^\top (2\mathbb{Z})^{-1}\big(\mathbb{X}+\mathbb{Z}\mathbb{H} \Phi\big)\Bigg]&\Sigma_q\\
   \star &-\mu\mathbb{I}
    \end{bmatrix}_{\Phi\in \mathcal{V}_{\phi}}\leq 0.
\end{equation}
\hrule   
\end{table*}
Now, we are ready to state the following theorem:
\begin{thm}\label{sec 3 thm 1}
If there exists a matrix $\mathbb{Z}$ under the form~\eqref{sec 4 EQN 11 Z} along with matrices $P_1=P_1^\top>0$, $R_1$ and $R_2$ of appropriate dimensions, and positive scalars $\epsilon,\beta$ such that the subsequent optimisation problem is solvable:
\begin{equation}
    \begin{split}\label{sec3 thm1 LMI}
        \text{minimise}~\mu~\text{subject to}&\\
        \begin{bmatrix}
            \Sigma_{11}&(\mathbb{M}+ \epsilon \mathbb{N})^\top&\big(\mathbb{X}+\mathbb{Z}\mathbb{H} \Phi\big)^\top&\Sigma_q\\
            \star &-2\epsilon P_1&\mathbb{O}&\mathbb{O}\\
            \star &\star&-2\mathbb{Z}&\mathbb{O}\\
            \star &\star&\star&-\mu\mathbb{I}
        \end{bmatrix}_{\Phi\in \mathcal{V}_{\phi}}\leq 0,
    \end{split}
\end{equation}
where 
\begin{equation}\label{sec3 thm1 LMI M}
    \mathbb{M}=\begin{bmatrix}
        \mathbb{O}\\ R_2^\top\bar{C}
    \end{bmatrix},
\end{equation} and
\begin{equation}\label{sec3 thm1 LMI N}
    \mathbb{N}=\begin{bmatrix}
        \bar{C}^\top R_1\\\mathbb{O}
    \end{bmatrix}.
\end{equation}
The terms $\Sigma_{11}$, $\mathbb{X}$, $\mathbb{H} $ and $\Phi$ are described in~\eqref{sec 4 eq 18 sigma11},~\eqref{sec 4 eq 18 X},~\eqref{sec 4 eq 18 H} and~\eqref{sec 4 eq 18 phi}, respectively.
Then, the estimation error dynamic~\eqref{sec 3 eq 14} is $\mathcal{H}_\infty$ asymptotically stable.
The gain matrices $\mathcal{K}$ and $L_2$ are computed as $\mathcal{K}=P_1^{-1}R_1^\top$ and $L_2=P_2^{-1}R_2^\top$. Figure~\ref{Fig Algo Parameter computation} depicts an algorithm for determining other observer parameters.
\end{thm}
\begin{proof}
Schur's complement of~\eqref{sec 4 eq 17 w con main} yields the LMI~\eqref{sec3 thm1 LMI}. From convexity principle\cite{boyd1994linear}, the estimation error dynamic~\eqref{sec 3 eq 14} fulfills $\mathcal{H}_\infty$ criterion~\eqref{sec 3 eq 15 H infinity} if the LMI~\eqref{sec3 thm1 LMI} is solved for all
$\Phi \in \mathcal{V}_{\phi}$.
\end{proof}
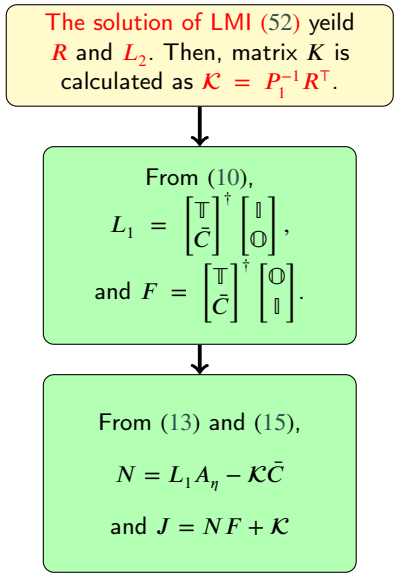
\begin{figure}[!h]
\centering
\begin{tikzpicture}
    \node [block3,node distance=2.5cm] (L2) {\textcolor{red}{The solution of LMI~\eqref{sec3 thm1 LMI}} yeild \textcolor{red}{$R$} and \textcolor{red}{$L_2$}. Then, matrix $K$ is calculated as \textcolor{red}{$\mathcal{K}=P_1^{-1}R^\top$}. };
\node[block4,below of=L2,node distance=2.5cm](L3)
{ From~\eqref{sec 3 eq 6 UIO con 1},
\\$L_1=\begin{bmatrix}
    \mathbb{T}\\\bar{C}
    \end{bmatrix}^{\dagger}\begin{bmatrix}
    \mathbb{I}\\ \mathbb{O}
    \end{bmatrix},$
    \\and~$F=\begin{bmatrix}
    \mathbb{T}\\\bar{C}
\end{bmatrix}^{\dagger}\begin{bmatrix}\mathbb{O}\\
    \mathbb{I}
    \end{bmatrix}$.
};
\node[block4, below of=L3,node distance=3cm](L4) 
{From~\eqref{sec 3 eq 6 UIO con 2 K} and~\eqref{sec 3 eq 6 UIO con 3},
$$N=L_1A_\eta-\mathcal{K} \bar{C}$$
$$\text{and}~J=NF+\mathcal{K}$$
};
    \path [line] (L2) -- (L3);
    \path [line] (L3) -- (L4);
\end{tikzpicture}
\caption{Algorithm to calculate observer parameters} 
\label{Fig Algo Parameter computation}
\end{figure}
\section{FAUIO synthesis for disturbance-affected nonlinear systems}\label{sec 4 NLDO with disturbance}
\subsection{Problem formulation}
The nonlinear system with linear outputs under the effect of external disturbances and faults is represented as follows:
\begin{align}\label{sec 4 eq 1.1}
\begin{split}
     \dot{x}  &= A x +B u
     +G g(x)+E_f f_a+E_1 \omega_1,\\
     y        &= C x+D_f f_s+D_1 \omega_2,
\end{split}
\end{align}
where $\omega_1\in \mathbb{R}^q_1$ and $\omega_2\in \mathbb{R}^q_2$ are external disturbance signals affecting both dynamics and outputs, respectively. $E_1\in \mathbb{R}^{n \times q_1}$ and $D_1\in \mathbb{R}^{p \times q_2}$ are known constant matrices. All other remaining variables and parameters of~\eqref{sec 4 eq 1} are the same as the one described in~\eqref{sec 3 eq 1}.
The function $f: \mathbb{R}^n \to \mathbb{R}^m$ is assumed to be globally Lipschitz, and it fulfils~\eqref{sec 3 eq 2}. The system~\eqref{sec 4 eq 1.1} is modified as:
\begin{align}\label{sec 4 eq 1}
\begin{split}
     \dot{x}  &= A x +B u
     +G g(x)+E_f f_a+E \omega,\\
     y        &= C x+D_f f_s+D \omega,
\end{split}
\end{align}
where, $\omega=\begin{bmatrix}
    \omega_1\\\omega_2
\end{bmatrix}$, $E=\begin{bmatrix}
    E_1 &\mathbb{O}
\end{bmatrix}$ and $E=\begin{bmatrix}
    \mathbb{O}&D_1
\end{bmatrix}$.

The system~\eqref{sec 4 eq 1} fulfil the assumptions~\ref{Asumption 1},~\ref{Asumption 2} and~\ref{Asumption 3}.
\begin{assum}\label{Asumption 4}
  In addition to these assumptions, it is presumed that the first-order derivative of disturbance $\omega$ is $\mathcal{L}_2$ bounded.
\end{assum}

Analogous to the previous section, the descriptor system for~\eqref{sec 4 eq 1} is illustrated as:
\begin{align}\label{sec 4 eq 3}
    \begin{split}
\mathcal{T} \dot{\zeta} &=
A_\zeta \zeta + B u +
G g(\mathcal{T}\zeta)+ E_f f_a+E\omega,\\
y&=\bar{C} \zeta+D\omega ,       
    \end{split}
\end{align}
where all variables and parameters remain consistent with the one specified in~\eqref{sec 3 eq 3}. 

The following FAUIO is deployed for the state estimation and fault estimation of the system~\eqref{sec 4 eq 1}:
\begin{equation}\label{sec 4 eq 5 obs}
    \begin{split}
\dot{\eta} &= N \eta + J y + L_1 B u + L_1 G g(\mathcal{T} \hat{\zeta})+L_1 E_f \hat{f}_a,\\
\hat{\zeta}&= \eta +F y,\\
\dot{\hat{f}}_a&=\beta L_2 (\tilde{y}+\dot{\tilde{y}}),\\
     \tilde{y}&=y-\bar{C}\hat{\zeta},
    \end{split}
\end{equation}
where all variables and parameters are described in~\eqref{sec 3 eq 5 obs}.
Let us consider augmented state estimation error $\tilde{\zeta}=\zeta-\hat{\zeta}$ and actuator fault estimation error $\tilde{f}_a=f_a-\hat{f}_a$. It yields:
\begin{equation}\nonumber
    \tilde{\zeta}=\zeta-\eta-F\bar{C}\zeta-F D \omega=(\mathbb{I}-F\bar{C})\zeta -\eta-F D \omega.
\end{equation}
It is assumed that the observer parameters $L_1$ and $F$ of~\eqref{sec 4 eq 5 obs} satisfy~\eqref{sec 3 eq 6 UIO con 1}, and it leads to:
\begin{equation}\label{sec 4 eq 7}
    \tilde{\zeta}=L_1 \mathcal{T}\zeta -\eta-F D \omega.
\end{equation}
Further,
\begin{equation}\small\nonumber
    \begin{split}
        \dot{\tilde{\zeta}}&=N \tilde{\zeta} +(L_1 A_\zeta-NL_1 \mathcal{T}-J\bar{C})\zeta+L_1 G\bigg( g(\mathcal{T}\zeta)- g(\mathcal{T} \hat{\zeta})\bigg)\\&+
L_1 E_f\tilde{f}_a+(L_1 E +(NF-J) D)\omega-FD\dot{\omega}.
    \end{split}
\end{equation}
From~\eqref{sec 3 eq 6 UIO con 2},~\eqref{sec 3 eq 6 UIO con 3} and~\eqref{sec 3 eq 7 g tilde}, we obtain:
\begin{equation}\small\label{sec 4 eq 8 er zeta final}
\begin{split}
  \dot{\tilde{\zeta}}&=(L_1 A_\zeta-\mathcal{K}\bar{C}) \tilde{\zeta} +\bigg(\sum_{i,j=1}^{m,\bar{n}} g_{ij}L_1 G\mathcal{H}_{ij} H_i\mathcal{T}  \bigg)\Tilde{\zeta}\\&+
L_1 E_f\tilde{f}_a+(L_1 E -\mathcal{K} D)\omega-FD\dot{\omega}.  
\end{split}
\end{equation}
Further, the fault estimation error dynamic is calculated as:
\begin{equation}\nonumber
        \dot{\tilde{f}}_{a}=-\beta L_2 \bar{C} \Tilde{\zeta} -\beta L_2 \bar{C} \dot{\Tilde{\zeta}}+\dot{f}_a-\beta L_2 D \omega -\beta L_2 D \dot{\omega}.
\end{equation}
It yields:
\begin{equation}\label{sec 4 eq 9 er fa final}
        \dot{\tilde{f}}_{a}=-\beta L_2 \bar{C} \Tilde{\zeta} -\beta L_2 \bar{C} \dot{\Tilde{\zeta}}+\dot{f}_a-\beta L_2 D \omega -\beta L_2 D \dot{\omega}.
\end{equation}
Through the utilisation of~\eqref{sec 4 eq 8 er zeta final} and~\eqref{sec 4 eq 9 er fa final}, one can deduce
\begin{equation}\label{sec 4 eq 10}
\begin{split}
\mathcal{T}_e\dot{e}&=\mathcal{A}_e e+\Bigg[\sum_{i,j=1}^{m,\bar{n}}  \begin{bmatrix}
    L_1 G\mathcal{H}_{ij}\\\mathbb{O}
    \end{bmatrix}\begin{bmatrix}
    g_{ij}H_i\mathcal{T}&\mathbb{O}
\end{bmatrix}\Bigg]e\\&+\begin{bmatrix}
L_1 E-\mathcal{K}D&-FD&\mathbb{O}\\
-\beta L_2D&-\beta L_2D&\mathbb{I}_{a_1}
    \end{bmatrix}
    \begin{bmatrix}
      \omega \\\dot{\omega}\\\dot{f}_a 
    \end{bmatrix},
\end{split}
\end{equation}
where $\mathcal{T}_e=\begin{bmatrix}
\mathbb{I}_{n_{\text{new}}}&\mathbb{O}\\\beta L_2 \bar{C} & \mathbb{I}_{a_1}
\end{bmatrix}~\text{and}~\mathcal{A}_e= \begin{bmatrix}
    L_1 A_\zeta-\mathcal{K}\Bar{C} & L_1 E_f\\-\beta L_2 \bar{C} & \mathbb{O}
    \end{bmatrix}$. It is easy to obtain: $\mathcal{T}_e^{-1}=\begin{bmatrix}
\mathbb{I}_{n_{\text{new}}}&\mathbb{O}\\-\beta L_2 \bar{C} & \mathbb{I}_{a_1}
\end{bmatrix}$.
\\Similar to~\eqref{sec 3 eq 11}, we can rewrite~\eqref{sec 4 eq 10} in the following manner:
\begin{equation}\small\label{sec 4 eq 11}
\begin{split}
\dot{e}&=\mathcal{T}_e^{-1}\mathcal{A}_e e+\Bigg[\sum_{i,j=1}^{m,\bar{n}} \mathcal{T}_e^{-1} \begin{bmatrix}
    L_1 G\mathcal{H}_{ij}\\\mathbb{O}
    \end{bmatrix}\begin{bmatrix}
    g_{ij}H_i\mathcal{T}&\mathbb{O}
\end{bmatrix}\Bigg]e\\&+\mathcal{T}_e^{-1}\begin{bmatrix}
L_1 E-\mathcal{K}D&-FD&\mathbb{O}\\
-\beta L_2D&-\beta L_2D&\mathbb{I}_{a_1}
    \end{bmatrix}
    \begin{bmatrix}
      \omega \\\dot{\omega}\\\dot{f}_a 
    \end{bmatrix}.
\end{split}
\end{equation}
Further, one can achieve:
\begin{equation}\scriptsize\nonumber
\begin{split}
\mathcal{T}_e^{-1}&\begin{bmatrix}
L_1 E-\mathcal{K}D&-FD&\mathbb{O}\\
-\beta L_2D&-\beta L_2D&\mathbb{I}_{a_1}
    \end{bmatrix}\\
    &=\begin{bmatrix}
    L_1 E-\mathcal{K}D&-FD&\mathbb{O}\\
  -\beta L_2\bar{C} (L_1 E-\mathcal{K}D)-\beta L_2 D&\beta L_2\bar{C}FD-\beta L_2 D&\mathbb{I}_{a_1}
    \end{bmatrix}
\end{split}    
\end{equation}
Let us introduce: $\bar{\omega}=\begin{bmatrix}
      \omega^\top &\dot{\omega}^\top &\dot{f}_a^\top
    \end{bmatrix}^\top$ and
\begin{equation}\scriptsize\label{sec 4 eq 12}
\mathbb{E}_\omega=\begin{bmatrix}
    L_1 E-\mathcal{K}D&-FD&\mathbb{O}\\
  -\beta L_2\bar{C} (L_1 E-\mathcal{K}D)-\beta L_2 D&\beta L_2\bar{C}FD-\beta L_2 D&\mathbb{I}_{a_1}
    \end{bmatrix}.    
\end{equation}
From~\eqref{sec 3 eq 12},~\eqref{sec 3 eq 13} and~\eqref{sec 4 eq 12}, the error dynamic~\eqref{sec 4 eq 11} is altered as:
\begin{equation}\scriptsize\label{sec 4 eq 14}
\dot{e}=\tilde{\mathcal{A}}_e e+\Bigg[\sum_{i,j=1}^{m,\bar{n}} \begin{bmatrix}
   L_1 G\mathcal{H}_{ij}\\
    -\beta L_2 \bar{C} L_1 G \mathcal{H}_{ij}
    \end{bmatrix}\begin{bmatrix}
    g_{ij}H_i\mathcal{T}&\mathbb{O}
\end{bmatrix}\Bigg]e+\mathbb{E}_\omega
    \bar{\omega}.
\end{equation}
Due to the fact that $\bar{C}$ is a full column matrix, conditions~\eqref{sec 3 rank con 1} and~\eqref{sec 3 rank con 2} are fulfilled. Hence, the necessary conditions for the existence of FAUIO~\eqref{sec 4 eq 5 obs} are satisfied.

The objective of this section is to deduce the observer parameters such that the error dynamic~\eqref{sec 4 eq 14} fulfills the following $\mathcal{H}_\infty$ criterion:
\begin{equation}\label{sec 4 eq 15 H infinity}
  ||e||_{\mathcal{L}^{n_{a_1}}_2}  \leq \sqrt{\nu ||e_0||^2+\mu  ||\bar{\omega}||^2_{\mathcal{L}^{a_1+2q}_2}},
\end{equation}
where $\mu,\nu > 0$. 
\subsection{LMI synthesis}

We will consider the quadratic Lyapunov function~\eqref{sec 3 eq LY fun} for $\mathcal{H}_\infty$ stability analysis of~\eqref{sec 4 eq 14}. The derivative of $V(e)$ along the trajectories of~\eqref{sec 4 eq 14} is illustrated as follows:
\begin{equation}\small\label{sec 4.1 eq 15 vdot}
    \begin{split}
    \dot{V}(e)&=   e^\top \Bigg[\Tilde{\mathcal{A}}_e^\top P+P\Tilde{\mathcal{A}}_e\Bigg]e+e^\top P \mathbb{E}_\omega
    \bar{\omega}+\bar{\omega}^\top \mathbb{E}_\omega^\top P e
    \\&+e^\top \Bigg[\sum_{i,j=1}^{m,\bar{n}}
P\begin{bmatrix}
   L_1 G\mathcal{H}_{ij}\\
    -\beta L_2 \bar{C} L_1 G \mathcal{H}_{ij}
    \end{bmatrix}\begin{bmatrix}
    g_{ij}H_i\mathcal{T}&\mathbb{O}
\end{bmatrix}\\&+\bigg(P\begin{bmatrix}
   L_1 G\mathcal{H}_{ij}\\
    -\beta L_2 \bar{C} L_1 G \mathcal{H}_{ij}
    \end{bmatrix}\begin{bmatrix}
    g_{ij}H_i\mathcal{T}&\mathbb{O}
\end{bmatrix}\bigg)^\top\Bigg]e.
    \end{split}
\end{equation}
From~\cite{zemouche2009unified}, one can say that the error dynamic~\eqref{sec 4 eq 14} fulfils the $\mathcal{H}_\infty$ criterion~\eqref{sec 4 eq 15 H infinity} if the ensuing inequality is true:
\begin{equation}\label{sec 4.1 eq 16}
  \mathcal{W} \triangleq  \dot{V}(e)+||e||^2-\mu ||\bar{\omega}||^2 \leq 0.
\end{equation}
By using~\eqref{sec 4.1 eq 15 vdot} and~\eqref{sec 4.1 eq 16}, $\mathcal{W}\leq 0$ if
\begin{equation}\label{sec 4.1 eq 17 w con}
    \begin{bmatrix}
\mathcal{A}_n &\mathcal{A}_{q} \\
    \star &-\mu \mathbb{I}
\end{bmatrix}\leq 0,
\end{equation}
where $\mathcal{A}_n$ is described in~\eqref{sec 4 eq 17 w con term 1} and
\begin{equation}\scriptsize\label{sec 4.1 eq 17 w con term 2}
\begin{split}
\mathcal{A}_{q}&=P\mathbb{E}_\omega\\&=\begin{bmatrix}
    P_1 L_1 E-P_1 \mathcal{K}D&-P_1 FD&\mathbb{O}\\
  -P_2 L_2\bar{C} (L_1 E-\mathcal{K}D)-P_2 L_2 D&P_2 L_2\bar{C}FD-P_2 L_2 D& \beta^{-1} P_2
    \end{bmatrix}.   
\end{split}
\end{equation}
Further, we can express $\mathcal{A}_q$ in the following manner:
\begin{equation}\scriptsize\label{sec 4.1 eq 17 w con term 2.1}
\begin{split}
\mathcal{A}_{q}&=\underbrace{\begin{bmatrix}
    P_1 L_1 E-P_1 \mathcal{K}D&-P_1 FD&\mathbb{O}\\
  -P_2 L_2\bar{C}L_1 E-P_2 L_2 D&P_2 L_2\bar{C}FD-P_2 L_2 D&\beta^{-1}P_2
\end{bmatrix}}_{\mathcal{A}^1_{{q}}}\\&+\underbrace{\begin{bmatrix}
    \mathbb{O}&\mathbb{O}&\mathbb{O}\\
  -P_2 L_2\bar{C} \mathcal{K}D&\mathbb{O}&\mathbb{O}
\end{bmatrix}}_{\mathcal{A}^2_{{q}}} 
\end{split}    
\end{equation}
Through the utilisation of~\eqref{sec 4 eq 17 w con term 1 fi} and~\eqref{sec 4.1 eq 17 w con term 2.1}, one can deduce~\eqref{sec 4.1 eq 17 w con main1} and 
\begin{equation}\small\label{sec 4.1 eqn sigma2}
    \Sigma_q=\begin{bmatrix}
    P_1 L_1 E-R_1^\top D&-P_1 FD&\mathbb{O}\\
  -R_2^\top\bar{C}L_1 E-R_2^\top D&R_2^\top\bar{C}FD-R_2^\top D&\beta^{-1} P_2
\end{bmatrix}.
\end{equation}
\begin{table*}[!ht]
  \begin{equation}\label{sec 4.1 eq 17 w con main1}
\begin{split}
  \begin{bmatrix}
   \Sigma_{11}+ (\mathbb{M}+ \epsilon \mathbb{N})^\top (-2\epsilon P_1)^{-1}(\mathbb{M}+ \epsilon \mathbb{N})  +\Bigg[\big(\mathbb{X}+\mathbb{Z}\mathbb{H} \Phi\big)^\top (2\mathbb{Z})^{-1}\big(\mathbb{X}+\mathbb{Z}\mathbb{H} \Phi\big)\Bigg]&\Sigma_q\\
   \star &-\mu\mathbb{I}
    \end{bmatrix}_{\Phi\in \mathcal{V}_{\phi}}+\begin{bmatrix}
    \mathbb{O}&\mathcal{A}^2_{{q}}\\
    \star&\mathbb{O}
    \end{bmatrix}\leq 0,  
\end{split}
\end{equation}
\hrule
\begin{equation}\scriptsize\label{sec 4.1 eq 17 w con main}
\begin{split}
  \begin{bmatrix}
   \Sigma_{11}+ (\mathbb{M}+ \epsilon \mathbb{N})^\top (-2\epsilon P_1)^{-1}(\mathbb{M}+ \epsilon \mathbb{N})  +\Bigg[\big(\mathbb{X}+\mathbb{Z}\mathbb{H} \Phi\big)^\top (2\mathbb{Z})^{-1}\big(\mathbb{X}+\mathbb{Z}\mathbb{H} \Phi\big)\Bigg]&\Sigma_q\\
   \star &-\mu\mathbb{I}
    \end{bmatrix}_{\Phi\in \mathcal{V}_{\phi}}+  (\mathbb{U}+\delta P_1 \mathbb{V})^\top(2\delta P_1)^{-1} 
(\mathbb{U}+\delta P_1 \mathbb{V}) \leq 0&.
\end{split}
\end{equation}
\hrule
\end{table*}
Further,
\begin{equation}\scriptsize
    \begin{split}
 \begin{bmatrix}
    \mathbb{O}&\mathcal{A}^2_{{q}}\\
    \star&\mathbb{O}
\end{bmatrix}&=\underbrace{\begin{bmatrix}
    \mathbb{O}\\
    \begin{bmatrix}
        \mathbb{O}\\-P_2 L_2 \bar{C}
    \end{bmatrix}
\end{bmatrix}}_{\mathbb{U}_1^\top} \overbrace{P_1^{-1}\begin{bmatrix}
    \begin{bmatrix}
        P_1 \mathcal{K}D & \mathbb{O}& \mathbb{O}
    \end{bmatrix} & \mathbb{O}
\end{bmatrix}}^{\mathbb{V}_1} + \mathbb{V}_1^\top \mathbb{U}_1    
    \end{split}
\end{equation}
By implementing~\eqref{L 3.1.1}, we get:
\begin{equation}
 \begin{bmatrix}
    \mathbb{O}&\mathcal{A}^2_{{q}}\\
    \star&\mathbb{O}
\end{bmatrix} \leq (\mathbb{U}_1+\delta P_1 \mathbb{V}_1)^\top(2\delta P_1)^{-1} 
(\mathbb{U}_1+\delta P_1 \mathbb{V}_1).
\end{equation}
Let us define $\mathbb{U}^\top=\begin{bmatrix}
    \mathbb{O}\\
    \begin{bmatrix}
        \mathbb{O}\\-R_2^\top \bar{C}
    \end{bmatrix}
\end{bmatrix}$ and $\mathbb{V}^\top=\begin{bmatrix}
    \begin{bmatrix}
         D^\top R_1 \\ \mathbb{O}\\\mathbb{O}
    \end{bmatrix}\\ \mathbb{O}
    \end{bmatrix}$.
\\Thus,
\begin{equation}
 \begin{bmatrix}\label{sec 4.1 eq 17 w con AQ1}
    \mathbb{O}&\mathcal{A}^2_{{q}}\\
    \star&\mathbb{O}
\end{bmatrix} \leq (\mathbb{U}+\delta P_1 \mathbb{V})^\top(2\delta P_1)^{-1} 
(\mathbb{U}+\delta P_1 \mathbb{V}).
\end{equation}
Therefore,~\eqref{sec 4.1 eq 17 w con main1} is reformulated into the condition~\eqref{sec 4.1 eq 17 w con main}.

After establishing the fundamental framework, we can define the following theorem.

\begin{thm}\label{sec 4 thm 1}
Let us consider the matrix $\mathbb{Z}$ under the form of~\eqref{sec 4 EQN 11 Z}. The estimation error dynamic~\eqref{sec 4 eq 14} is $\mathcal{H}_\infty$ asymptotically stable if there exist the matrices $P_1=P_1^\top>0 $  $P_2=P_2^\top>0$,  $R_1$ and $R_2$ of appropriate dimensions such that the ensuing optimization problem is solvable:  
\begin{equation}
    \begin{split}\scriptsize\label{sec4 thm2 LMI}
    \text{minimise}~\mu~\text{subject to}&\\
        \begin{bmatrix}
            \Sigma_{11}&(\mathbb{M}+ \epsilon \mathbb{N})^\top&\big(\mathbb{X}+\mathbb{Z}\mathbb{H} \Phi\big)^\top&\Sigma_q&\Sigma^1_{q}\\
            \star &-2\epsilon P_1&\mathbb{O}&\mathbb{O}&\mathbb{O}\\
            \star &\star&-2\mathbb{Z}&\mathbb{O}&\mathbb{O}\\
            \star &\star&\star&-\mu\mathbb{I}&\mathbb{O}\\
            \star &\star&\star&\star&
            -2\delta P_1\end{bmatrix}_{\Phi\in \mathcal{V}_{\phi}}\leq 0&,
    \end{split}
\end{equation}
where 
\begin{equation}
   \Sigma^1_{q}= (\mathbb{U}+\delta P_1 \mathbb{V})^\top,
\end{equation}
and $\Sigma_q$ is described in~\eqref{sec 4.1 eqn sigma2}. All other remaining variables and parameters remain consistent with those specified in Theorem~\ref{sec 3 thm 1}. The computation of other observer parameters is illustrated in Figure~\ref{Fig Algo Parameter computation}.
\end{thm}
\begin{proof}
One can easily deduce LMI~\eqref{sec4 thm2 LMI} by implementing Schur Lemma on~\eqref{sec 4.1 eq 17 w con main}.
Through the utilisation of convexity principle\cite{boyd1994linear}, the estimation error dynamic~\eqref{sec 4 eq 14} satisfies $\mathcal{H}_\infty$ criterion~\eqref{sec 4 eq 15 H infinity} if the LMI~\eqref{sec4 thm2 LMI} is solved for all
$\Phi \in \mathcal{V}_{\phi}$.
\end{proof}
\section{Validating the observer performance}\label{sec num examples}
This section is devoted to the evaluation of the performance of the proposed methodology. 
In order to accomplish this goal, an application of a single-link flexible joint robot is used. In the first part, through the utilisation of this application, the performance of the LMI~\eqref{sec3 thm1 LMI} is validated. Later on, the effectiveness of LMI~\eqref{sec3 thm1 LMI} is showcased with the same application.

\subsection{Under the presence of sensor and fault attack}
A single-link flexible joint robot model is represented in the form of~\eqref{sec 3 eq 1} with the ensuing parameters:
\\$A=\begin{bmatrix}
    0   &  1 &    0 &  0\\
   -48.6 & -1.25 &48.6 &0\\
      0   & 0 &    0&   1\\
    19.5   &0   &-19.5 &0
\end{bmatrix}$, $B=\begin{bmatrix}
    0\\21.6\\0\\0
\end{bmatrix}$, $G=\begin{bmatrix}
    0\\0\\0\\-3.33
\end{bmatrix}$, $C=\begin{bmatrix}
    1&0&0&0\\
    0&1&0&0\\
    0&0&1&1
\end{bmatrix}$, $E_f=B$, $D_f=\begin{bmatrix}
   1\\0  \\0 
\end{bmatrix}$,  and $g(x)=sin(x_3)$. Let us presume $H_1=\begin{bmatrix}
    0 & 0& 1 & 0\\
   -1 & 1 &0  &1\\
    1 & 0 &0& -1\\
    0 &-1 &0 & 0
\end{bmatrix}$. 
\\Thus, $n=4$, $a_1=1$, $a_2=1$,  $m=1$ and $\bar{n}=4$.

The actuator and sensor faults that occurred in the system are illustrated as follows:
\begin{align}\nonumber   
f_a(t)&=\begin{cases}
0, & t\in [0,15)\\
3 \sin (0.5t) + 2 \cos (5t), & t\in [15,30)\\
0 &  t\in [30,50).\end{cases}
\end{align}
\begin{align}\nonumber   
f_s(t)=\begin{cases}
0, & t\in [0,5)\\
5-0.5(t-20), & t\in [5,35)\\
0 &  t\in [35,50).\end{cases}
\end{align}
\begin{figure}[!h]
    \centering
\includegraphics[width=\linewidth]{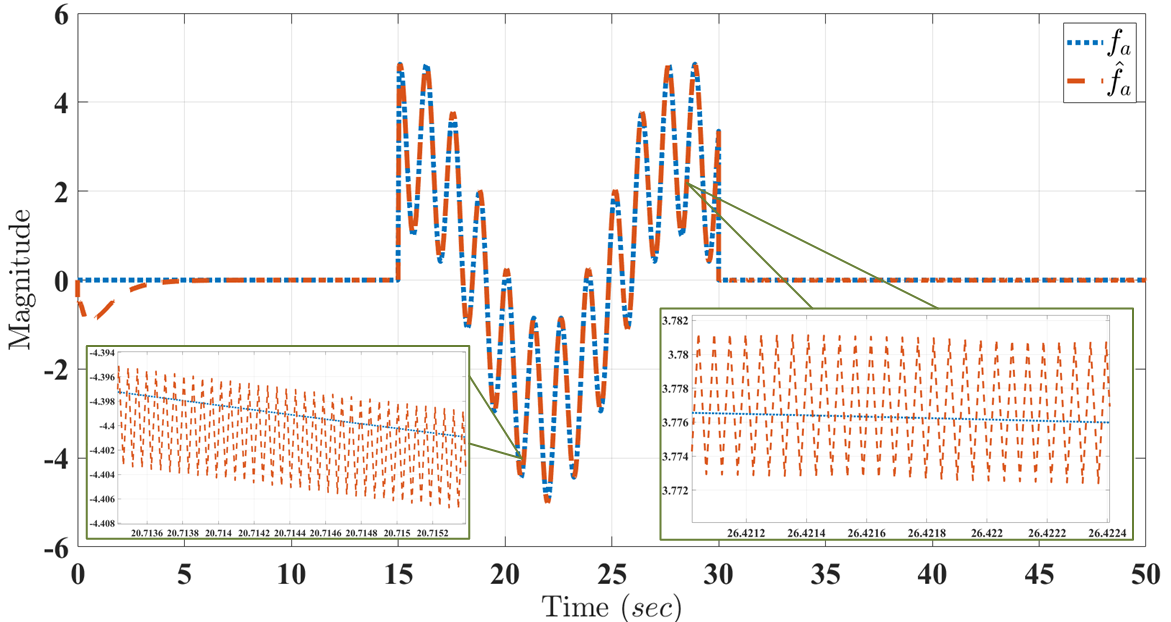}
    \caption{Graph of estimation of actuator fault vector ($f_a$)}
    \label{fig 1 est fa}
\end{figure}
Further, LMI~\eqref{sec3 thm1 LMI} is solved in MATLAB toolbox to compute the parameters of the unknown input observer~\eqref{sec 3 eq 5 obs} by considering $\beta=100,\epsilon=0.1$.
The obtained value is showcased in Table~\ref{tab 1}. It infers that the LMI approach of~\cite{ref_UIO_est_noiseless} is not suitable in this case. Thus, the effectiveness of the proposed LMI method is emphasised.
\begin{table}[!h]
    \centering
    \caption{Encapsulating LMI results}
    \label{tab 1}
    \begin{tabular}{|c|c|c|}
    \hline
     LMIs    &LMI~\eqref{sec3 thm1 LMI}&~\cite[LMI~(14)]{ref_UIO_est_noiseless}  \\
     \hline
         $\sqrt{\mu}$&$2.5399\times 10^{-4}$ &Infeasible\\ \hline    
    \end{tabular}
\end{table}
\\The obtained parameters are outlined as follows:
$$~\mathcal{K}=\begin{bmatrix}
    -0.0261  &  0.6630&    4.1125\\
    0.8767   &11.3468 & -23.8897\\
    0.3344   & 0.4541 & -17.2965\\
   -0.6458   &-0.7730 &  36.4162\\
    1.2139   &-0.6630 &  -4.1125
\end{bmatrix};~L_2=\begin{bmatrix}
    0.0000 \\  66.5084  \\183.6139
\end{bmatrix}^\top.$$
One can deduce the following matrices by utilising the algorithm shown in Figure~\ref{Fig Algo Parameter computation}:
\begin{equation}\small\nonumber
   N=\begin{bmatrix}
0.0261 & 0.3370 &-4.1125&-4.1125&  0.0261\\
-25.1767 &-11.9718&48.1897&23.8897 &-0.8767\\
-6.8344&-0.4541 &23.7965&  17.9632& -0.3344\\
13.6458& 0.7730&-49.4162&-36.7496 & 0.6458\\
-1.2139 &-0.3370 & 4.1125 &4.1125 &  -1.2139
\end{bmatrix};
\end{equation}
\begin{equation}\nonumber
J=\begin{bmatrix}
    -0.0000 &   0.8315  &  1.3708\\
    0.0000  &  5.3609   & 0.1368\\
    0.0000  &  0.2271   &-3.3766\\
   -0.0000  & -0.3865   & 7.6943\\
    0.0000  & -0.8315   &-1.3708
\end{bmatrix}.   
\end{equation}

\begin{figure}[!h]
    \centering
\includegraphics[width=\linewidth]{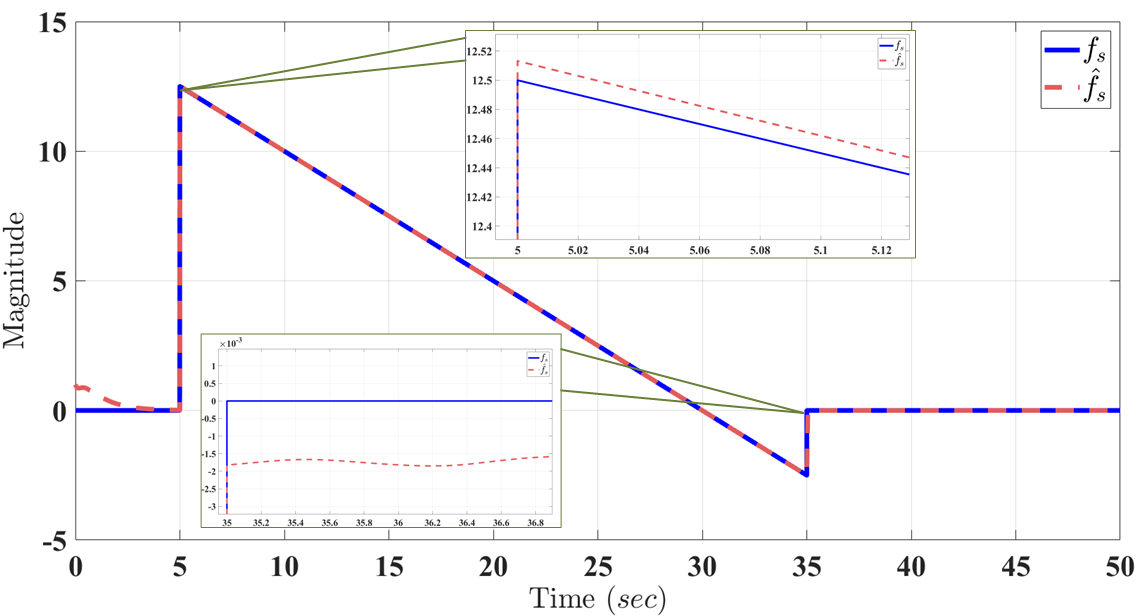}
    \caption{Plot of estimation of sensor fault vector ($f_s$)}
    \label{fig 1 est fs}
\end{figure}
With all these aforementioned matrices, the unknown input observer~\eqref{sec 3 eq 5 obs} is implemented in the MATLAB environment for estimating the faults and states simultaneously. The graphical representation of the estimated and actual actuator fault is portrayed in Figure~\ref{fig 1 est fa}. It depicts the accuracy of the estimation of actuator faults. It is mainly because of the selection of learning rate $\beta$. Figure~\ref{fig 1 est fs} illustrates the plot of the estimated and actual sensor fault signal. Further, the estimation errors of $f_a$ and $f_s$ are displayed in Figure~\ref{fig 1 est error fa fs}. The spikes shown in the graph of $\tilde{f}_a$ indicate sudden changes in actuator fault signals. In addition to this, one can notice that the reconstruction of the actuator faults, as well as the sensor fault, is sufficiently accurate, and reliable. Hence, the performance of the observer~\eqref{sec 3 eq 5 obs} is validated.

\begin{figure}[!h]
    \centering
\includegraphics[width=\linewidth]{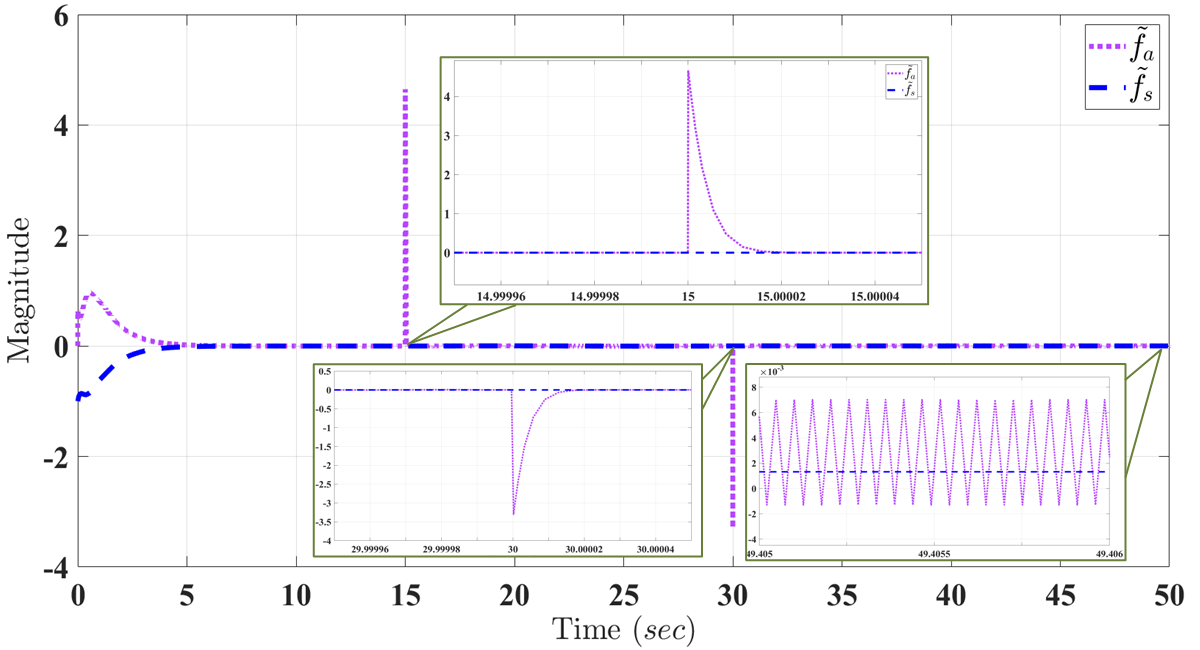}
    \caption{Plot of estimation error of $f_a$ and $f_s$}
    \label{fig 1 est error fa fs}
\end{figure}

\subsection{Scenario of disturbance-affected system}
In order to analyse the performance of the proposed observer~\eqref{sec 4 eq 5 obs}, the earlier-stated example is considered. The system dynamics and measurements are presumed to be affected by external disturbances/noise $\omega_1=0.2\sin (10t)$ and $\omega_2=0.1 \sin(10t)$ for $\forall t\in [0,50]$, respectively, along with
$E_1=\begin{bmatrix}
   0 \\0 \\0 \\1 
\end{bmatrix}$ and
$D_1=\begin{bmatrix}
   0.1 \\-0.02  \\0 
\end{bmatrix}$. Thus, the state-space form described in~\eqref{sec 4 eq 1} is utilised to represent this example by considering $\omega=\begin{bmatrix}
    \omega_1\\\omega_2
\end{bmatrix}$, $E=\begin{bmatrix}
    E_1 &\mathbb{O}
\end{bmatrix}$ and $E=\begin{bmatrix}
    \mathbb{O}&D_1
\end{bmatrix}$.

Let us presume $\beta=100,~\epsilon=0.0112,~\delta=5$.
Then, one can compute the following parameters by solving LMI~\eqref{sec4 thm2 LMI} in the MATLAB toolbox:
$\sqrt{\mu}=0.0324;~L_2=\begin{bmatrix}
    6.3355 \\  30.7302\\    4.7777
\end{bmatrix}^\top;~\mathcal{K}=\begin{bmatrix}
    1.5587 &   7.9111&   12.5406\\
   28.1685 & 144.2130&  124.2935\\
  -26.3414 &-116.4108& -203.3290\\
   52.6611 & 232.7129&  408.5888\\
  -69.1253 &-350.9884& -710.9153
\end{bmatrix}$. 
Similar to the previous segment, we deduce the subsequent matrices: 
\begin{equation}\scriptsize\nonumber
\begin{split}
   N&=\begin{bmatrix}
-1.5587& -6.9111&  -12.5406 & -12.5406&   -1.5587\\
  -52.4685 &-144.8380  &-99.9935& -124.2935 & -28.1685\\
   19.8414 & 116.4108&  209.8290 & 203.9957  & 26.3414\\
  -39.6611 &-232.7129 &-421.5888& -408.9222 & -52.6611\\
   69.1253 & 349.9884&  710.9153&  710.9153&   69.1253
\end{bmatrix};\\J&=\begin{bmatrix}
   0.0000  &  4.4555  &  4.1802\\
    0.0000  & 71.7940 &  49.5312\\
   -0.0000 & -58.2054 & -65.3875\\
    0.0000 & 116.3565 & 131.7518\\
   -0.0000 &-175.9942 &-236.9718
\end{bmatrix}.   
\end{split}
\end{equation}
\begin{table}[!h]
    \centering
    \caption{Summarising the value of $\sqrt{\mu}$ for different approaches}
    \label{sec5.2 tab 1 mu comparision}
    \begin{tabular}{|c|c|c|c|}
    \hline
     LMIs    &LMI~\eqref{sec4 thm2 LMI}&~\cite[LMI~(66)–(67) and
(69)]{gao2022fast}&~\cite{gao_2016_estimation_fault}  \\
     \hline
         $\sqrt{\mu}$&$0.0324$ &$1.736$&$1.76$\\ \hline    
    \end{tabular}
\end{table}
Table~\ref{sec5.2 tab 1 mu comparision} illustrates the comparison of the noise attenuation level obtained from the different LMI approaches. It highlights that the proposed LMI condition provides a better noise compensation than the existing methods.

Further, by utilising these above-mentioned matrices, the unknown input observer~\eqref{sec 4 eq 5 obs} is deployed in the MATLAB environment to estimate the faults and states simultaneously.

\begin{table}[!h]
    \centering
    \caption{Settling time (in$~{sec}$) for estimation error $\tilde{f}_a$ in several cases}
    \label{sec5.2 tab 1}
    \begin{tabular}{|c|c|c|c|c|}
\hline
Error & Different methods & \textbf{Case 1} & \textbf{Case 2} & \textbf{Case 3} \\ \hline
\multirow{3}{*}{$\tilde{f}_a$} & FAUIO~\eqref{sec 4 eq 5 obs}  &0.00015  &0.00013  &0.00011  \\ \cline{2-5} 
                  &\cite{gao2022fast}  &0.099  &0.2396 &0.21  \\ \cline{2-5} 
        &\cite{gao_2016_estimation_fault}  & 2.791 & 2.04 & 0.9\\ \hline
\end{tabular}
\end{table}

For the lucidity of the presentation, the authors divided the estimation analysis under several fault scenarios into the ensuing cases:
\begin{enumerate}[I]
    \item \textbf{Case 1}: Impulsive faults
    \\ Let us consider that the aforementioned system is subjected to the ensuing impulsive faults:
    
    \begin{align}\nonumber   
f_a(t)&=\begin{cases}
0, & t\in [0,20)\\
0.1(t-10), & t\in [20,20.1)\\
0 &  t\in [20.1,50).\end{cases}
\end{align}
\begin{align}\nonumber   
f_s(t)=\begin{cases}
0, & t\in [0,30)\\
0.1(t-10), & t\in [30,30.1)\\
0 &  t\in [30.1,50).\end{cases}
\end{align}
\begin{figure}[!h]
    \centering
\includegraphics[width=\linewidth]{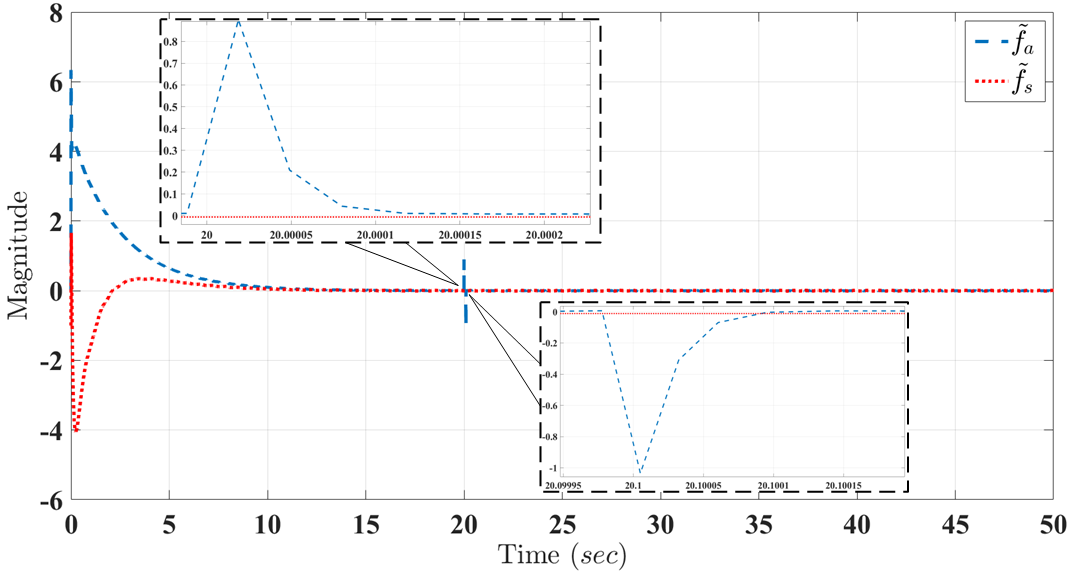}
    \caption{Graphical representation of estimation error $\tilde{f}_a$ and $\tilde{f}_s$ in \textbf{Case 1}}
    \label{sec 5.2 fig1 case 1}
\end{figure}
Through the utilisation of the proposed observer~\eqref{sec 4 eq 5 obs} in MATLAB, the faults are estimated. The plot of the obtained estimation error of faults is illustrated in Figure~\ref{sec 5.2 fig1 case 1}. It depicts the asymptotic convergence of both error $\tilde{f}_a$ and $\tilde{f}_s$. In addition to this, it infers that the proposed observer estimates both faults rapidly as well as accurately, which can be shown in the sequel.
\begin{figure}[!h]
    \centering
\includegraphics[width=\linewidth]{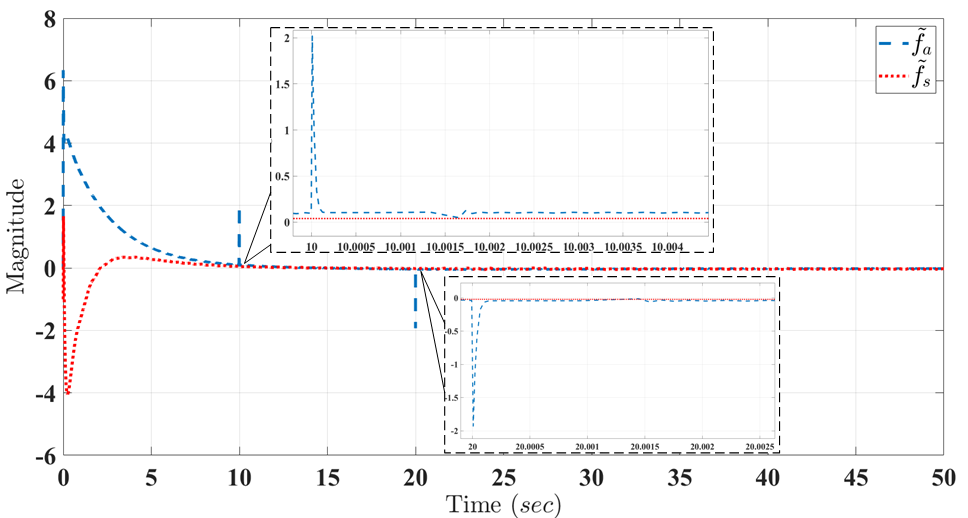}
    \caption{Graphical representation of estimation error $\tilde{f}_a$ and $\tilde{f}_s$ in \textbf{Case 2}}
    \label{sec 5.2 fig1 case 2}
\end{figure}
\item \textbf{Case 2}: Abrupt faults
    \\ The previously indicated system is affected by the following fault signals:
    \begin{align}\nonumber   
f_a(t)&=\begin{cases}
0, & t\in [0,10)\\
2, & t\in [10,20)\\
0 &  t\in [20,50).\end{cases}
\\   
f_s(t)&=\begin{cases}
0, & t\in [0,30)\\
2, & t\in [30,35)\\
0 &  t\in [35,50).\end{cases}\nonumber
\end{align}
Analogous to the previous case, the plot of the estimation error of faults is obtained in MATLAB Simulink, and showcased in Figure~\ref{sec 5.2 fig1 case 2}. The asymptotic convergence of both error $\tilde{f}_a$ and $\tilde{f}_s$ is shown in Figure~\ref{sec 5.2 fig1 case 2}. 
Thus, the performance of the proposed observer in this case is validated. The comment regarding estimation accuracy and settling time of estimation error is provided in the next part of this segment.
\begin{figure}[!h]
    \centering
\includegraphics[width=\linewidth]{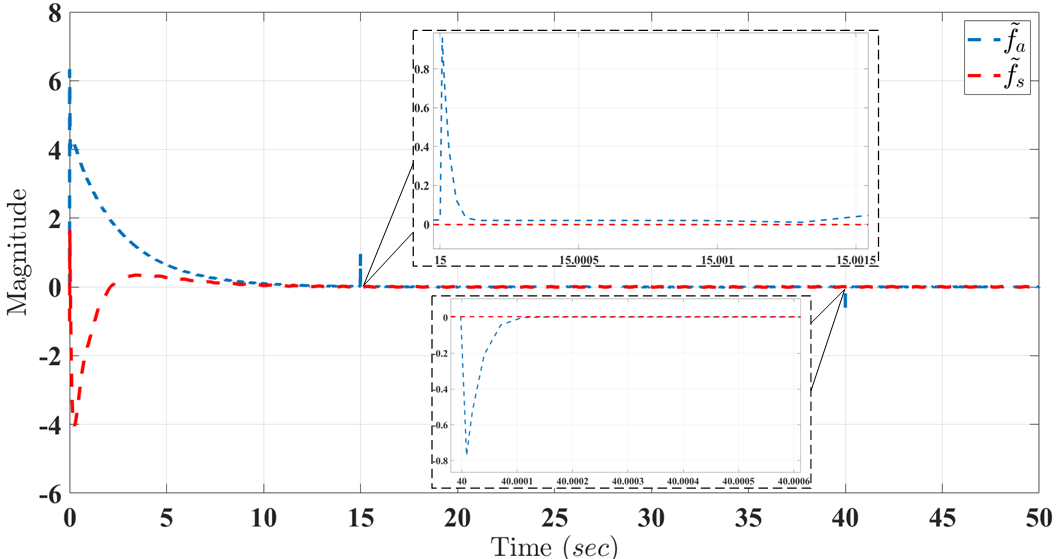}
    \caption{Graphical representation of estimation error $\tilde{f}_a$ and $\tilde{f}_s$ in \textbf{Case 3}}
    \label{sec 5.2 fig1 case 3}
\end{figure}
    \item \textbf{Case 3} Sinusoidal faults
    \\Let us presume that the earlier-stated system is under the impact of the following fault signals:
    \begin{align}  
f_a(t)&=\begin{cases}
0, & t\in [0,15)\\
\sin (0.5t)+0.2\cos (5t), & t\in [15,40)\\
0 &  t\in [40,50).\end{cases}\nonumber 
\\f_s(t)&=\begin{cases}
0, & t\in [0,15)\\
\sin (0.5t)+0.2\cos (5t),, & t\in [15,35)\\
0 &  t\in [35,50).\end{cases}\nonumber 
\end{align}
Similar to the preceding case, the behaviour of the estimation error of faults achieved in MATLAB Simulink is showcased in Figure~\ref{sec 5.2 fig1 case 3}. It portrays the asymptotic convergence of $\tilde{f}_a$ and $\tilde{f}_s$ with fast convergence and high accuracy. 
Thus, the performance of the proposed observer in~\textbf{Case 3} is verified.
\end{enumerate}

Further, the settling time required for the convergence of $\tilde{f}_a$ for each case is summarised in Table~\ref{sec5.2 tab 1}, which highlights the proposed observer provides a rapid convergence of estimation error $\tilde{f}_a$. The root mean square of the fault estimation error (RMSE) in each case is outlined in Table~\ref{sec5.2 tab 2}, and it showcased the accuracy of the estimation over the existing approaches. Thus, the proposed methodology performs efficiently in the case of impulsive faults.
\begin{table}[!h]
\caption{RMSE values of estimation errors $\tilde{f}_a$ and $\tilde{f}_s$ in several cases}
\label{sec5.2 tab 2}
\begin{tabular}{|c|c|c|c|c|}
\hline
Error             & Different methods & \textbf{Case 1} & \textbf{Case 2} & \textbf{Case 3} \\ \hline
\multirow{3}{*}{$\tilde{f}_a$} &FAUIO~\eqref{sec 4 eq 5 obs}  & 0.0076               &      0.0340            & 0.0189     \\ \cline{2-5}&  \cite{gao2022fast}     &           0.02      &       0.0596          &          0.04    \\ \cline{2-5}& \cite{gao_2016_estimation_fault}                  & 0.196                &  0.116               &      0.121           \\ \hline
\multirow{3}{*}{$\tilde{f}_s$} &FAUIO~\eqref{sec 4 eq 5 obs}  & 0.0078                &    0.0295            & 0.0120                \\ \cline{2-5} &  \cite{gao2022fast}&          0.018       &  0.0315               & 0.05                \\ \cline{2-5}  & \cite{gao_2016_estimation_fault}  &  0.110               &   0.105              &  0.102               \\ \hline
\end{tabular}
\end{table}

\section{Conclusion}\label{sec conclusion}
This letter is dedicated to the establishment of an FAUIO-based estimation approach for the rapid and accurate reconstruction of states and faults of nonlinear systems. It is achieved by the development of a novel LMI condition that aids in determining the parameters of the FAUIO such that the estimation errors of states and faults are asymptotically stable. The derived LMI condition is obtained by combining the reformulated Lipschitz property, a new variant of Young inequality, and a well-known LPV approach. The formulated LMI criterion is less conservative than the existing ones, which is highlighted using a numerical example. Further, this novel FAUIO approach is deployed for the disturbance-affected nonlinear system, and a new LMI condition is designed to estimate the faults and states with optimal noise attenuation. The performance of the developed FAUIO observer is validated by using the application of a single-link robotic arm manipulator in the MATLAB environment. From a future perspective, this methodology can be extended for discrete-time nonlinear systems.
In addition to this, the proposed approach can be combined with the existing fault-tolerant controller (FTC) methods to design a novel FAUIO-based FTC methodology.

\section{ Acknowledgement}
The authors would like to express their gratitude to Dr Ali Zemouche and Dr. Maraoune Alma, University of Lorraine, CRAN, Nancy, France for their technical advice and detailed conversations. The authors are also thankful to Dr N. M. Singh, EED, VJTI, Mumbai, India,  for his technical support during manuscript preparation.

 \bibliographystyle{elsarticle-num}
 \bibliography{cas-refs}

\printcredits




\end{document}